\documentclass[10pt,final,letterpaper]{IEEEtran}

\usepackage[utf8]{inputenc}
\usepackage{amsmath}
\usepackage{amsthm}
\usepackage{mathtools}
\usepackage{amsfonts}
\usepackage{hyperref}
\usepackage{url}
\usepackage{color}
\usepackage{tikz}
\usepackage{stackengine}
\usepackage{pgfplots}
\pgfplotsset{compat=1.16}
\usetikzlibrary{spy}
\usetikzlibrary{intersections}
\usepackage{xargs}
\usepackage[noadjust]{cite}
\usepackage{titlesec}
\usepackage{glossaries}
\usepackage{arydshln}
\usepackage{xcolor}
\definecolor{revcolor}{rgb}{.7,0,0}
\newcommand{\rev}[1]{{#1}}
\newcommand{\revspace}[1]{}

\usepackage{caption}

\pgfplotsset{
    discard if/.style 2 args={
        x filter/.code={
            \edef\tempa{\thisrow{#1}}
            \edef\tempb{#2}
            \ifx\tempa\tempb
                
            \fi
        }
    },
    discard if not/.style 2 args={
        x filter/.code={
            \edef\tempa{\thisrow{#1}}
            \edef\tempb{#2}
            \ifx\tempa\tempb
            \else
                
            \fi
        }
    }
}

\titlespacing*{\subsection}{0pt}{6pt}{3pt}
\newcommand\blankfootnote[1]{%
  \let\svthefootnote\thefootnote%
  \let\thefootnote\relax\footnotetext{#1}%
  \let\thefootnote\svthefootnote%
}

\makeatletter
\newcommand\footnoteref[1]{\protected@xdef\@thefnmark{\ref{#1}}\@footnotemark}
\makeatother

\newcommand{\Real}{\mathop{\mathrm{Re}}\nolimits}

\newcommand{\E}{{\mathbb{E}}}

\newcommand{\eps}{{\varepsilon}}        

\newcommand{\fr}{{\mathcal F}}

\newcommand{\sign}{\operatorname{sign}}

\newcommand{\NN}{{\mathbb{N}}}

\newcommand{\complex}{\mathbb{C}}
\newcommand{\reals}{\mathbb{R}}

\newcommand{\Bs}{{\mathcal{B}}}

\newcommand{\Ss}{{\mathcal{S}}}

\newcommand{\abs}[1]{{\bigl\lvert #1\bigr\rvert}}

\newcommand{\subgaussnorm}[1]{{\tau\left({#1}\right)}}
\newcommand{\subexpnorm}[1]{{\theta\left({#1}\right)}}
\newtheorem{theorem}{Theorem}

\newtheorem{lemma}{Lemma}
\newtheorem{cor}{Corollary}

\newtheorem{remark}{Remark}
\newtheorem{definition}{Definition}

\newcommand{\numUsers}{{K}}
\newcommand{\indexUsers}{{k}}
\newcommand{\numChannelUses}{{M}}
\newcommand{\indexChannelUses}{{m}}
\newcommand{\channelVector}{{H}}
\newcommand{\noiseVector}{{N}}
\newcommand{\normalizedMessages}{{X}}
\newcommand{\messagesMatrix}{{Q}}

\newcommand{\rxVector}{{Y}}
\newcommand{\randomBaseVector}{{R}}
\newcommand{\transpose}[1]{{{#1}^T}}
\newcommand{\fadingTransformMatrix}{{A}}
\newcommand{\noiseTransformMatrix}{{B}}

\newcommand{\vectorNorm}[1]{\| {#1} \|_2}
\newcommand{\generalIndexOne}{{i}}
\newcommand{\generalIndexTwo}{{j}}
\newcommand{\hwEffectiveMatrix}{{C}}
\newcommand{\Expectation}{{\mathbb{E}}}
\newcommandx{\absolute}[3][1=\left, 2=\right]{#1 | #3 #2 |}
\newcommand{\tail}{{\varepsilon}}
\newcommand{\operatornorm}[1]{\| {#1} \|}
\newcommand{\frobeniusnorm}[1]{\| {#1} \|_F}
\newcommand{\Probability}{{\mathbb{P}}}
\newcommand{\powerconstraint}{{P}}
\newcommand{\mgfVariable}{\lambda}

\newcommand{\hwDiagonalSum}{{\Sigma_1}}
\newcommand{\hwNondiagonalSum}{{\Sigma_2}}
\newcommand{\auxVariable}{{c}}
\newcommand{\generalUpperBound}{{\tilde{L}}}
\newcommand{\generalSumUpperBound}{{\tilde{F}}}
\newcommand{\identityMatrix}{{\mathbf{id}}}
\newcommand{\trace}{{\mathrm{tr}}}

\newcommand{\matrixOneEntry}[2]{{E_{{#1}, {#2}}}}
\newcommand{\unitaryFadingApprox}{A_U}

\newcommand{\fadingApproxError}{{\eta}}
\newcommand{\uncorrelatedApprox}{A_i}
\newcommand{\errorProbOperatorNormTerm}{L}
\newcommand{\errorProbFrobeniusNormTerm}{F}
\newcommand{\errorProbDependenceTerm}{D}
\newcommand{\transmitPower}{{a}}
\newcommand{\generalRandomVector}{{X}}
\newcommand{\generalDeterministicVector}{{x}}
\newcommand{\gaussianRandomVector}{{g}}
\newcommand{\subgaussBound}{{K}}
\newcommand{\rvDimension}{{n}}
\newcommand{\rvDimIndex}{{k}}
\newcommand{\rvDimIndexAlt}{{\ell}}
\newcommand{\generalTransformMatrix}{{A}}
\newcommand{\mvNormal}[2]{{\mathcal{N}({#1},{#2})}}
\newcommand{\innerprod}[2]{{\left\langle{#1}, {#2}\right\rangle}}
\newcommand{\euclidnorm}[1]{\left\| {#1} \right\|_2}
\newcommand{\generalTransformMatrixSingular}{{s}}
\newcommand{\gaussianRandomVectorTransformed}{{h}}
\newcommand{\subgaussdefInfVar}{{t}}
\newcommand{\subgaussdefUnitVector}{{a}}
\newcommand{\sphere}[1]{{S^{{#1}-1}}}
\newcommand{\operatornormBound}{A_\mathrm{op}}
\newcommand{\frobeniusnormBound}{A_\mathrm{F}}

\newcommand{\identityFunction}{{\mathbf{id}}}
\newcommand{\indicator}[1]{\mathbf{1}_{#1}}
\newcommand{\generalmatrixOne}{{X}}
\newcommand{\generalMatrixTwo}{{Y}}
\newcommand{\subgaussvariable}{{\xi}}
\newcommand{\generalVector}{{v}}

\newcommand{\absoluteConstantOne}{c}
\newcommand{\absoluteConstantTwo}{C}

\newcommand{\fadingBlockLength}{\beta}

\newcommand{\middletonImpulsive}{\mathbf{A}}
\newcommand{\middletonRatio}{\mathbf{\Gamma}}

\newcommand{\middletonIntermediateRV}{\mathbf{m}}
\newcommand{\noisePower}{P_N}

\newacronym{svm}{SVM}{Support Vector Machine}
\newacronym{ml}{ML}{Machine Learning}
\newacronym{ota}{OTA}{Over-the-Air}
\newacronym{fl}{FL}{Federated Learning}
\newacronym{vfl}{VFL}{Vertical Federated Learning}
\newacronym{dfa}{DFA}{Distributed Function Approximation}
\newacronym{csi}{CSI}{channel state information}
\newacronym{tdma}{TDMA}{Time Division Multiple Access}
\newacronym{awgn}{AWGN}{Additive White Gaussian Noise}
\newacronym{irs}{IRS}{Intelligent Reflective Surface}
\newacronym{mimo}{MIMO}{multiple-input multiple-output}

\begin{document}
\title{Over-The-Air Computation in Correlated Channels}

\author{
  \IEEEauthorblockN{
    Matthias Frey\IEEEauthorrefmark{1},
    Igor Bjelakovi\'c\IEEEauthorrefmark{2} %
    and %
    S\l awomir~Sta\'{n}czak\IEEEauthorrefmark{1}\IEEEauthorrefmark{2}\\ %
  } %
  \IEEEauthorblockA{
    \IEEEauthorrefmark{1}Technische Universität Berlin, Germany\\
    \IEEEauthorrefmark{2}Fraunhofer Heinrich Hertz Institute, Berlin, Germany 
  }%
  \vspace*{-2em}
}

\maketitle
\blankfootnote{
Part of this work was presented in~\cite{bjelakovic2019distributed} at the 57th Annual Allerton Conference on Communication, Control, and Computing, Sept. 24-27, 2019, Allerton Park and Retreat Center, Monticello, IL, USA.

Part of this work has been presented in~\cite{frey2021over} at the IEEE 2020 Information Theory Workshop (ITW).

This work was supported by the German Research Foundation (DFG) within their priority programs SPP 1798 ``Compressed Sensing in Information Processing'' and SPP 1914 "Cyber-Physical Networking", as well as under grant STA 864/7.

This work was partly supported by a Nokia University Donation. We gratefully acknowledge the support of NVIDIA Corporation with the donation of the DGX-1 used for this research.
}
\begin{abstract}
\gls{ota} computation is the problem of computing functions of distributed data without transmitting the entirety of the data to a central point. By avoiding such costly transmissions, \gls{ota} computation schemes can achieve a better-than-linear (depending on the function, often logarithmic or even constant) scaling of the communication cost as the number of transmitters grows. Among the most common functions computed \gls{ota} are linear functions such as weighted sums. In this work, we propose and analyze an analog \gls{ota} computation scheme for a class of functions that contains linear functions as well as some nonlinear functions such as $p$-norms of vectors. We prove error bound guarantees that are valid for fast-fading channels and all distributions of fading and noise contained in the class of sub-Gaussian distributions. This class includes Gaussian distributions, but also many other practically relevant cases such as Class A Middleton noise and fading with dominant line-of-sight components. In addition, there can be correlations in the fading and noise so that the presented results also apply to, for example, block fading channels and channels with bursty interference. We do not rely on any stochastic characterization of the distributed arguments of the \gls{ota} computed function; in particular, there is no assumption that these arguments are drawn from identical or independent probability distributions. Our analysis is nonasymptotic and therefore provides error bounds that are valid for a finite number of channel uses. OTA computation has a huge potential for reducing communication cost in applications such as Machine Learning (ML)-based distributed anomaly detection in large wireless sensor networks. We illustrate this potential through extensive numerical simulations.
\end{abstract}

\section{Introduction}
\label{sec:intro}

Due to the rapid increase in the number of wireless devices and the continuous emergence of new wireless applications, the available radio frequency spectrum is becoming increasingly scarce. Since the demand for wireless connectivity is expected to continue at an even faster pace in the future, we will face the spectrum crunch if we do not develop adequate solutions, especially in the sub-6 GHz frequency range. Exploiting new frequency bands beyond the 6 GHz limit alone cannot solve the spectrum crunch problem given the  massive deployment of IoT devices that is planned in coming years. In addition, energy efficiency is increasingly becoming a key issue. Indeed, when gathering data from massively distributed IoT devices, the scaling laws for capacity and energy are relevant and need to be improved, otherwise system performance may be severely degraded~\cite{gupta2000capacity}. Significant improvements in scaling laws can be achieved by abandoning the philosophy of strict separation of communication and application-specific computation. OTA computation schemes, such as the one considered in this paper, compute a function of data distributed in a wireless network at a central receiver without reconstructing the data in its entirety. The problem of efficient function computation is also becoming increasingly important in the context of research on future 6G networks~\cite{shi2020communication}. One example of promising OTA applications is the training of machine learning (ML) models in wireless (sensor) networks, since it only requires the computation of some features of the sensor data -- in contrast, the data of the individual sensors do not need to be decoded. In fact, if the goal for a receiver is to compute a function 
$f:\reals^\numUsers\to\reals$ of some $\numUsers$ variables, rather than fully reconstructing all individual variables, then a strategy based on the use of orthogonal channels is in general highly suboptimal. In particular, depending on the function being considered, the OTA computation can achieve \emph{a more favorable scaling law for the capacity} than traditional separation-based approaches~\cite{nazer2007computation}.

\subsection{Prior Work}
\label{sec:prior_work}

Analog uncoded approximation of functions first appeared in~\cite{gastpar2003source}. This work assumes known source distributions and additive white Gaussian noise channels without fading for the achievability theorems and the class of approximated functions is constrained to the linear case. \cite{goldenbaum2013robust,goldenbaum2013harnessing,goldenbaum2014nomographic} have picked up this idea and extended to a more general class of functions. These works consider the noiseless case as well as the case with noise, but without fast fading, providing asymptotic error bounds. 

The authors of~\cite{nazer2007computation} have introduced distributed computation of functions with coding and proposed applications in network coding. Several more recent works (e.g., \cite{zhan2009mimo,ordentlich2011practical,nazer2011compute,nazer2016expanding,goldenbaum2016harnessing}) have refined, expanded and applied this idea.
These works focus on the application of network coding and therefore the case in which the same function (an addition in a finite field) is to be computed repeatedly. The coding approach used ensures computation of the discrete functions with an arbitrarily small error as long as the computation rate is not too high.

\cite{goldenbaum2014channel} presents pre- and post-processing schemes and an asymptotic analysis for the approximation of functions over a fast fading channel with noise in the case of multiple receive antennas. The work covers the case of no instantaneous~\gls{csi} at the transmitter or receiver as well as two different types of instantaneous~\gls{csi} at the transmitter. Reference~\cite{kiril} povides pre- and post-processing schemes with asymptotic guarantees for approximating linear functions over fast-fading channels with \gls{csi} at the transmitter. The authors of~\cite{liu2020over} derive theoretical bounds on the mean squared error in \gls{ota} function computation in a fast-fading scenario with \gls{csi} at the transmitter. The work also considers the case of multi-antenna transceivers and provides empirical results along with the theoretical bounds. In~\cite{dong2020blind}, the authors approach the \gls{ota} computation problem without explicit channel estimation under the assumption of slow fading. The work also considers intersymbol interference and provides an asymptotic theoretical analysis as well as numerical results.

The direct application of \gls{ota} computation techniques to distributed gradient descent has received a lot of attention recently since this can be used to solve the empirical risk minimization problem for \gls{ml} models such as neural networks in the case of distributed training data without having to collect the training data at a central point. \cite{amiri2020machine,ahn2019wireless} propose to extend the \gls{fl} paradigm~\cite{mcmahan2017communication,konecny2016federated} to make use of \gls{ota} computation over wireless channels and provide theoretical analysis along with empirical results to this end. There are also extensions of this idea to channels with fading channel information at either the transmitter or receiver~\cite{zhu2019broadband,amiri2020federated,yang2020federated,sery2020analog,sun2019energy}, often taking additional aspects into consideration such as differential privacy~\cite{seif2020wireless} and multi-antenna scenarios~\cite{amiri2019collaborative}. 

\subsection{Overview of Nomographic Functions}
\label{sec:nomographic}
Generally, we expect functions of the form
\begin{equation*}
f(s_1, \ldots, s_K)= F\left(  \sum_{k=1}^K f_k(s_k)     \right ),
\end{equation*}
called \emph{nomographic functions},
to be amenable to distributed approximation over a wireless channel in
which a superposition of signals results in a noisy sum of the
transmitted signals to arrive at the receiver, and in fact it turns
out that every multivariate real function $f$ has such a
representation~\cite{buck1976approximate}. However, even extremely
weak noise in the individual components can have an unpredictable
impact on the overall error if such representations are
used. Therefore, it is necessary to introduce additional requirements
on $f_1, \dots, f_K$ and $F$. It is known~\cite{kolmogorov1957representation} that every continuous function mapping from $[0,1]^K$ to $\reals$ can be represented as a sum of no more than $2K+1$ nomographic functions with continuous representations. Another result worth noting in this context is that nomographic functions with continuous representations are nowhere dense in the space of continuous functions~\cite{buck1982nomographic} and thus the representation as a sum of nomographic functions is really necessary.
However, even with suitable continuous representations available, it is hard to control the impact of the channel noise. We therefore identify a subclass of the nomographic functions for which the impact of noise and fading can be quantitatively controlled. This class contains many functions of practical interest such as $p$-norms and weighted sums, which are of particular importance in \gls{ml} applications.

\subsection{Wireless Channel Models}
\label{sec:intro-channel-models}
 In theoretical works, it is of particular importance that the considered channel model is sufficiently general so that the assumptions made are met in relevant practical scenarios. The commonly considered Gaussian fading channel is an approximation that is often adopted because it is relatively easy to treat and is quite close to reality in many scenarios of interest.

However, it is well known and has been confirmed in extensive measurement campaigns~\cite{middleton1993elements,blackard1993measurements,middleton1999non} that there are many natural and artificial sources of noise that do not conform to the assumption of being i.i.d. Gaussian such as automobile ignition, power line emissions, atmospheric disturbances or interfering wireless communications~\cite{middleton1993elements}.

Moreover, common arguments that the fast fading in wireless channels is Gaussian employ the Central Limit Theorem~\cite[Section 2.4.1]{yin2016propagation} and therefore assume a large number of multipath components, which is not always the case. In scenarios with limited mobility, it is also possible that the fast fading realizations are not independent between channel uses. In the case of the \gls{ota} computation scheme proposed in~\cite{bjelakovic2019distributed}, it can deteriorate performance and therefore, it is desirable to be able to quantify this deterioration.

In this work, we therefore consider a channel model that encompasses a large class of possible probability distributions of fading and noise, the class of sub-Gaussian distributions. The analysis provided is valid also for fading and noise exhibiting an arbitrary correlation structure, with practically useful bounds in many relevant cases. In Section~\ref{sec:subgauss}, we define precisely what sub-Gaussian means and in Section~\ref{sec:systemmodel-discussion}, we give examples of practically relevant cases that are covered by our system model assumptions.

\subsection{Contributions of this Paper}
\label{sec:contributions}
The main contributions of this work are summarized as follows:

\begin{enumerate}
 \item We propose a scheme for analog \gls{ota} computation without instantaneous \gls{csi} that does not assume any probability distribution on the data. Therefore, it works equally well for independently distributed data as it does for arbitrarily correlated values.
 \item We provide a nonasymptotic analysis of the error that is valid for a large class of possible distributions of fading and noise. The distributions covered include many types of non-Gaussian distributions and distributions with correlation over time. The scheme itself does not depend on the correlation structure, but the error analysis requires different tools than in the uncorrelated case.
 \item Our scheme can deal with a larger class of functions than previous works with theoretically proven bounds. In particular, we do not require linearity of the function to be approximated, which is, e.g., demonstrated by the fact that we can compute $p$-norms \gls{ota}.
 \item We propose applications of our scheme to the \gls{ota} computation of both regressors and classifiers in \gls{fl} and validate the proposed \gls{ota} computation schemes and the envisioned applications in \gls{ml} with extensive numerical simulations for the case of a binary classification problem.
\end{enumerate}

%
%
\section{System Model}

%
%
\subsection{Sub-Gaussian Random Variables}
\label{sec:subgauss}
We begin with a short overview of the relevant definitions and properties of sub-Gaussian random variables.
More on this topic can be found in Section~\ref{app:sub-exp-sub-gauss} of the Supplement and in \cite{buldygin,wainwright,vershynin}.

For a random variable $X$, we define\footnote{Note that other norms on the space of sub-Gaussian random variables that appear in the literature are equivalent to $\subgaussnorm{\cdot}$ (see, e.g.,~\cite{buldygin}). The particular definition we choose here matters, however, because we want to derive results in which no unspecified constants appear.}
\begin{multline}\label{eq:sub-gauss-norm-def-first}
\subgaussnorm{X} :=  \inf \Big\{t > 0: \forall \lambda \in \reals \\ \mathbb{E}\exp \left( \lambda (X - \mathbb{E} X) \right)   \le \exp \left( \lambda^2 t^2 / 2 \right)      \Big\}.
\end{multline}
 $X$ is called a sub-Gaussian random variable if $\subgaussnorm{X}<\infty $. The function $\subgaussnorm{\cdot}$ defines a semi-norm on the set of sub-Gaussian random variables \cite[Theorem 1.1.2]{buldygin}, i.e., it is absolutely homogeneous, satisfies the triangle inequality, and is non-negative. $ \subgaussnorm{X}=0$ does not necessarily imply $X=0$ unless we identify random variables that are equal almost everywhere.
Examples of sub-Gaussian random variables include Gaussian and bounded random variables.
%
%
\subsection{Channel and Correlation Model}
\label{sec:systemmodel}
\begin{figure}
\begin{tikzpicture}
\coordinate                      (inK)          at (0,0);
\coordinate                      (in2)          at (0,2);
\coordinate                      (in1)          at (0,3);
\node[rectangle,draw]            (encK)         at (1,0) {$E_K^M$};
\node[rectangle,draw]            (enc2)         at (1,2) {$E_2^M$};
\node[rectangle,draw]            (enc1)         at (1,3) {$E_1^M$};
\node[rectangle,draw]            (dec)          at (6.5,1.5) {$D_M$};
\coordinate                      (out)          at (8,1.5);
\node                            (vdots)        at (.9,1) {\Shortstack{. . . . . .}};

\node[circle,inner sep=0pt,draw] (multK)        at (2.5,0)    {$\times$};
\node[circle,inner sep=0pt,draw] (mult2)        at (2.5,2)    {$\times$};
\node[circle,inner sep=0pt,draw] (mult1)        at (2.5,3)  {$\times$};
\node[inner sep=0pt]             (fadingK)      at (2.5,.5)   {$H_K$};
\node[inner sep=0pt]             (fading2)      at (2.5,2.5)  {$H_2$};
\node[inner sep=0pt]             (fading1)      at (2.5,3.5)  {$H_1$};
\coordinate                      (cornerK)      at (3.4,0);
\coordinate                      (corner2)      at (3.4,2);
\coordinate                      (corner1)      at (3.4,3);
\node[circle,inner sep=0pt,draw] (plus)         at (4,1.5) {$+$};
\node[circle,inner sep=0pt,draw] (plusnoise)    at (5,1.5) {$+$};
\node[inner sep=0pt]             (noise)        at (5,2) {$N$};

\draw[->] (inK) -- (encK) node[midway,above] {$s_K$};
\draw[->] (in2) -- (enc2) node[midway,above] {$s_2$};
\draw[->] (in1) -- (enc1) node[midway,above] {$s_1$};

\draw[->] (encK) -- (encK-|multK.west) node[midway,above] {$x_K^M$};
\draw[->] (enc2) -- (enc2-|mult2.west) node[midway,above] {$x_2^M$};
\draw[->] (enc1) -- (enc1-|mult1.west) node[midway,above] {$x_1^M$};

\draw[->] (plusnoise.east) -- (dec) node[midway,above] {$Y^M$};
\draw[->] (dec) -- (out) node[midway,above] {$\tilde{f}$};

\draw[->] (fadingK) -- (multK);
\draw[->] (fading2) -- (mult2);
\draw[->] (fading1) -- (mult1);

\draw[->] (multK) -- (cornerK) -- (plus);
\draw[->] (mult2) -- (corner2) -- (plus);
\draw[->] (mult1) -- (corner1) -- (plus);

\draw[->] (plus) -- (plusnoise);
\draw[->] (noise) -- (plusnoise);
\end{tikzpicture}
\caption{System model.}
\label{fig:system}
\end{figure}
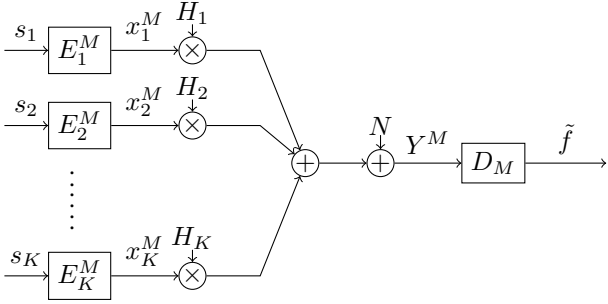

We consider the following channel model with $K$ transmitters and one receiver: For $m=1,\ldots, M$, the channel output at the $m$-th channel use is given by
\begin{equation}\label{eq:channel-model}
Y(m)=\sum_{k=1}^{K} H_k(m)x_k(m)+ N(m).
\end{equation}
Here and hereafter, the notation is defined as follows:
\begin{itemize}
\item $x_k(m)\in \complex$ is a transmit symbol. We assume a peak power constraint $|x_k(m)  |^2 \le P$ for $k=1,\ldots, K$ and $m=1,\ldots, M$.
\item $H_k(m)$, $k=1\ldots, K$, $m=1,\ldots, M$, are complex-valued random variables such that for every $m=1, \ldots ,M$ and $k=1,\ldots,K$, the real part $H_k^r (m)$ and the imaginary part $H_k^i (m)$  of $H_k(m)$ are sub-Gaussian random variables
with mean zero and variance $1$.
\item $N(m)$, $m=1,\ldots, M$, are complex-valued random variables. We assume that the real and imaginary parts $N^{r} (m),N^{i }(m)$ of $N(m)$  are sub-Gaussian random variables  with mean zero for $m=1,\ldots, M$.
\end{itemize}

In order to be able to apply a variation of the Hanson-Wright inequality as a tool, we give the formal description of our dependence model in terms of matrices and vectors with real entries.

We define 
\begin{align}
\label{eq:channelvector}
\channelVector := \transpose{(\channelVector(1), \dots, \channelVector(2\numChannelUses))}
\end{align}
where for $\indexChannelUses = 1, \dots, \numChannelUses$,
\begin{align*}
\channelVector(2\indexChannelUses-1) &:= (H_1^r(\indexChannelUses), \dots, H_K^r(\indexChannelUses))
\\
\channelVector(2\indexChannelUses) &:= (H_1^i(\indexChannelUses), \dots, H_K^i(\indexChannelUses)).
\end{align*}
So $\channelVector$ is the vector of all fading coefficients. Similarly, let 
\begin{align}
\label{eq:noisevector}
\noiseVector := \transpose{(N^r(1), N^i(1), \dots, N^r(\numChannelUses), N^i(\numChannelUses))}
\end{align}
be the vector of all the instances of additive noise. The dependence model we consider is such that there is a vector $\randomBaseVector$ of $(2\numUsers\numChannelUses + 2\numChannelUses)$ independent random variables with sub-Gaussian norm at most $1$ and matrices $\fadingTransformMatrix \in \reals^{2\numUsers\numChannelUses \times (2\numUsers\numChannelUses + 2\numChannelUses)}$ and $\noiseTransformMatrix \in \reals^{2\numChannelUses \times (2\numUsers\numChannelUses + 2\numChannelUses)}$ such that

\[
\channelVector = \fadingTransformMatrix \randomBaseVector,~~~
\noiseVector = \noiseTransformMatrix \randomBaseVector.
\]

Our correlation model captures both correlations between users and in the time domain, but the impact of these two types of correlation on the performance of the proposed scheme is different. For this reason, we need the following definition, which describes a scenario where there can be arbitrary correlation in the time domain but no harmful correlation between different users. In order to be as unrestrictive as possible, the following definition prohibits only correlations between different users at the same complex dimension of the same channel use.

\begin{definition}
\label{def:useruncorrelated}
We say that the fading is \emph{user-uncorrelated} if for every $\indexUsers_1 \neq \indexUsers_2$, $j \in \{i,r\}$ and $\indexChannelUses$, the random variables $H^j_{\indexUsers_1}(\indexChannelUses)$ and $H^j_{\indexUsers_2}(\indexChannelUses)$ are independent.
\end{definition}

In Section~\ref{sec:specialcases}, we discuss the impact of deviations from the user-uncorrelatedness assumption on the performance of the proposed method.

\begin{remark}
\label{remark:uncorrelatedApprox}
Obviously, user-uncorrelated fading can be characterized based on the form of $\fadingTransformMatrix$. If we write
\begin{align*}
\fadingTransformMatrix
=
\begin{pmatrix}
\fadingTransformMatrix^{(1)} \\
\vdots
\\
\fadingTransformMatrix^{(2\numChannelUses)}
\end{pmatrix},
\end{align*}
where for all $\indexChannelUses$, $\fadingTransformMatrix^{(\indexChannelUses)} \in \reals^{\numUsers \times (2\numChannelUses\numUsers + 2\numChannelUses)}$, then $\channelVector = \fadingTransformMatrix \randomBaseVector$ defines user-uncorrelated fading for all $\randomBaseVector$ iff each $\fadingTransformMatrix^{(\indexChannelUses)}$ has at most one nonzero entry per column. This is because $\channelVector(\indexChannelUses) = \fadingTransformMatrix^{(\indexChannelUses)}\randomBaseVector$ and therefore two or more nonzero entries in a column of $\fadingTransformMatrix^{(\indexChannelUses)}$ mean that two or more entries in $\channelVector(\indexChannelUses)$ depend on the same entry in $\randomBaseVector$.
\end{remark}

\subsection{Discussion of the System Model Assumptions}
\label{sec:systemmodel-discussion}
The most distinguishing feature of our channel model compared to assumptions made in prior work on \gls{ota} computation is that we generalize the distribution of the fading and noise to be sub-Gaussian and to feature arbitrary correlations. We argue that this guarantees robustness against departure from standard assumptions such as independent or Gaussian fading or noise. Regarding the chosen dependence model, we remark that in the case of $\randomBaseVector$ distributed i.i.d. standard Gaussian, this amounts to a standard representation of arbitrarily correlated (and thus arbitrarily interdependent) multivariate Gaussian vectors. Therefore, by replacing $\randomBaseVector$ with a vector of independent sub-Gaussian entries, we obtain a straightforward generalization of the Gaussian case, which specializes to arbitrary correlations (although in the non-Gaussian case, not to arbitrary stochastic dependence). In particular, the following important cases are specializations of our channel model:
\begin{itemize}
 \item Perfect Gaussian fast fading with i.i.d. bivariate Gaussian fading and additive white Gaussian noise.
 \item Scenarios where the number of multipath components is not sufficiently large for an appeal to the Central Limit Theorem to argue that the fading is complex Gaussian. For instance, if there is only one multipath component with a length that has small random variations, the resulting complex fading will have a distribution supported on a narrow annulus in the complex plane. Such a distribution is not Gaussian, but sub-Gaussian and therefore covered by our channel model.
 \item Scenarios where due to limited or slow movement of transmitters, receiver and environment, the independence assumption between subsequent realizations of the channel fading is not satisfied or where the diversity in the radio environment is so small that some of the users have correlated channels. One example of a widely used channel model that is a special case of the one considered in this paper is the block fading channel. In this case, we have a correlation of $1$ between fading realizations of the same block and $0$ between fading realizations in different blocks.
 \item Any of the above scenarios where in addition to thermal additive noise, we have interference from users outside the system. E.g., in the case of digital modulated signals, such interference consists of a sequence of transmitted constellation points, which is not Gaussian. Also, such interfering signals are inherently bursty in nature and therefore cannot be argued to be independent between different points in time. However, signals from realistic transmitters are always bounded in amplitude and therefore sub-Gaussian, which means that they are covered by our system model.
 \item Any of various types of artificial and natural interference that is not necessarily Gaussian, but limited in power. \cite{middleton1993elements,blackard1993measurements,middleton1999non} investigate various sources of non-Gaussian interference of this type, both correlated and uncorrelated over time, through theoretical modeling as well as extensive experiments and measurements confirming the theoretical models and the non-Gaussianity of the sources. \cite[Table 2.1]{middleton1993elements} enumerates several examples of this type of interference, including, e.g., solar radiation, automobile ignition and power line EM emissions.
\end{itemize}
The other main difference between our system model and the models used in existing works is that we do not make any assumption on a distribution of the input data at the transmitters. This means in particular that we cover arbitrary dependencies in the input data, which can be very important for applications, because the input data can depend, e.g., on local sensor readings recorded by devices or on training data collected for the same \gls{ml} problem. Therefore, in many relevant application scenarios, it is important that the transmission schemes employed are robust even to high, but unknown levels of correlations in the transmitted data.

\section{Problem Statement}

%
%
\subsection{Distributed Approximation of Functions}
Our goal is to approximate functions  $f: \Ss_1\times \ldots \times \Ss_K\to\reals$  in a distributed setting. The sets $\Ss_1,\ldots \Ss_K\subseteq \reals $ are assumed to be closed and endowed with their natural Borel $\sigma$-algebras 
$\Bs(\Ss_1),\ldots ,\Bs(\Ss_K)$, and we consider the  product $\sigma$-algebra $\Bs (\Ss_1)\otimes \ldots \otimes \Bs(\Ss_K)$ on the set $ \Ss_1\times \ldots \times \Ss_K $. Furthermore, the functions  $f: \Ss_1\times \ldots \times \Ss_K\to\reals$ under consideration are assumed to be measurable in what follows.

\begin{definition}
An admissible \gls{dfa} scheme for $f: \Ss_1\times \ldots \times \Ss_K\to\reals$ with $M$ channel uses, depicted in Fig.~\ref{fig:system}, is a pair $(E^M, D^M)$, consisting of:

\begin{enumerate}
\item A pre-processing function $E^M = (E_1^M, \dots, E_K^M)$, where each $E_k^M$ is of the form
\[E_k^M(s_k)=(E_k(m, s_k, U_k(m) ))_{m=1}^{M}\in \complex^{M} \]
where $U_k(1), \ldots, U_k(M)$ are random variables and the map
\[(s_k, t_1, \ldots, t_M)\mapsto  (E_k(m, s_k,  t_m ))_{m=1}^{M}\in \complex^{M}\]
is measurable for $s_k$ ranging over $\mathcal{S}_k$ and $t_1, \ldots, t_M$ ranging over the same sets as the random variables $U_k(1), \ldots, U_k(M)$.
The encoder $E_k^M$ is subject to the peak power constraint $ |E_k(m, s_k, U_k(m))  |^2 \le P$ for all $k=1,\ldots , K$ and $m=1,\ldots , M$.
\item A post-processing function $D^M$: The receiver is allowed to apply a measurable recovery function $D^M: \complex ^M\to \reals$ upon observing the output of the channel.
\end{enumerate}
\end{definition}

So in order to approximate $f$, the transmitters apply their pre-processing maps to
$(s_1,\ldots, s_K)\in \Ss_1\times \ldots \times \Ss_K$
resulting in $E_1^M (s_1), \ldots, E_K^M(s_K)$, which are sent to the receiver using the channel $M$ times. 
The receiver observes the output of the channel and applies the recovery map $D^M$. The whole process defines an estimate $\tilde{f}$ of $f$.

Let $\varepsilon > 0,\delta \in (0,1)$ and $f: \Ss_1\times \ldots \times \Ss_K\to\reals $ be given. We say that $f$ is $\varepsilon$-approximated after $M$ channel uses with confidence level $\delta$ if there is a \gls{dfa} scheme
$(E^M,D^M)$ such that the resulting estimate $\tilde{f}$ of $f$ satisfies
\begin{equation}\label{eq:eps-delta-approx}
\mathbb{P}( |    \tilde{f} (s^K)- f(s^K)       |\ge \eps      )\le \delta
\end{equation}
for all $s^K:= (s_1, \ldots , s_K) \in \Ss_1\times \ldots \times \Ss_K $.
Let $M(f, \varepsilon, \delta)$ denote the smallest number of channel uses such that there is an approximation scheme $(E^M, D^M)$ for $f$ satisfying (\ref{eq:eps-delta-approx}). We call $M(f, \varepsilon, \delta)$ the communication cost for approximating a function $f$
with accuracy $\varepsilon$ and confidence $\delta$.

%
%
%
\subsection{The class of functions to be approximated}
A measurable function  $f:\Ss_1\times \ldots \times \Ss_K \to \reals$ is called a \emph{generalized linear function} if there are bounded measurable functions $(f_k)_{k \in \{1, \ldots , K\}}$, with
$
f(s_1, \ldots , s_K)= \sum_{k=1}^K f_k(s_k),
$
for all $(s_1,\ldots , s_K)\in \Ss_1 \times \ldots \times \Ss_K$. The set of generalized linear functions from $\Ss_1 \times \ldots \times \Ss_K\to \reals$ is denoted by $\fr_{K, \textrm{lin}}$.
Our main object of interest will be the following class of functions.
%
%
\begin{definition}
\label{def:Fmon}
A measurable function $f: \Ss_1\times \ldots \times \Ss_K\to \reals$ is said to belong to $\fr_{\textrm{\textrm{mon}}}$ if there exist bounded and measurable functions
$(f_k)_{k \in \{1, \ldots , K\}}$, a measurable set $D\subseteq \reals$ with the property $f_1(\Ss_1)+\ldots + f_K(\Ss_K)\subseteq D$, a measurable function $F:D\to \reals$ such that
for all $(s_1, \ldots , s_K)\in \Ss_1\times \ldots \times \Ss_K$ we have
\begin{equation}\label{eq:nomographic-def}
f(s_1, \ldots, s_K)= F\left(  \sum_{k=1}^K f_k(s_k)     \right ),
\end{equation}
and there is a strictly increasing function $\Phi : [0, \infty) \to [0, \infty)$  with $\Phi(0)=0$ and
\begin{equation}\label{eq:monotone-domination}
| F(x)-F(y)|\le \Phi( | x-y | )
\end{equation}
for all $x,y \in D$. We call the function $\Phi$ an \emph{increment majorant} of $f$.
\end{definition}
Some examples of functions in $\fr_{\textrm{mon}}$ are:
\begin{enumerate}
\item Obviously, all $f\in\fr_{K, \textrm{lin}}$ belong to $ \fr_{\textrm{mon}}$.
\item For any $f\in\fr_{K, \textrm{lin}}$ and $B$-Lipschitz function $F:\reals \to \reals$ we have $F\circ f\in \fr_{\textrm{mon}}$ with
$\Phi: [0, \infty) \to [0, \infty)$, $x\mapsto Bx$.
\item If $f \in\fr_{K, \textrm{lin}}$ and $F$ is $(C, \alpha)$-Hölder continuous, i.e., for all $x,y$ in the domain of $F$,
$
\abs{F(x)-F(y)} \leq C\abs{x-y}^\alpha,
$
then $F \circ f \in \fr_{\textrm{mon}}$ with $\Phi: x \mapsto Cx^\alpha$.
\item For any $p\ge 1$ and $\Ss_1, \ldots , \Ss_K$ compact, $|| \cdot ||_p \in \fr_{\textrm{mon}}$. In this example we have $f_k(s_k)=| s_k |^p$, $k=1,\ldots , K$,
$F: [0, \infty) \to [0, \infty)$, $x \mapsto x^{\frac{1}{p}}$, and $F=\Phi$.\\
This can be seen as follows. We have to show that for all nonnegative $x,y\in \reals$ and $p\ge 1$ we have
\begin{equation}
\label{eq:norm-mon}
| x^{\frac{1}{p}}   -y^{\frac{1}{p}}   |\le |  x-y   |^{\frac{1}{p}}.
\end{equation}
We can assume w.l.o.g. that $x<y$ holds. Then since
\begin{equation*}
| x^{\frac{1}{p}}   -y^{\frac{1}{p}}   | = |y|^{\frac{1}{p}} \left( 1- \left(  \frac{x}{y}\right)^{\frac{1}{p}}     \right)      
\end{equation*}
it suffices to prove that for all $a\in [0,1]$ and $p\ge 1$ we have
$
1-a^{\frac{1}{p}}\le \left(  1-a     \right) ^{\frac{1}{p}}
$.
Now since
$
a^{\frac{1}{p}} +\left(  1-a     \right) ^{\frac{1}{p}}
\geq a + (1-a) = 1
$
for $a\in [0,1]$ and $p\ge 1$, we can conclude that (\ref{eq:norm-mon}) holds.
\end{enumerate}

Given a function $f \in \fr_{\textrm{mon}}$, we implicitly fix a representation (\ref{eq:nomographic-def}) and define the total spread of  the inner part of $ f\in  \fr_{\textrm{\textrm{mon}}}$ as
\begin{equation}\label{eq:total-inner-spread}
\bar{\Delta}(f):=\sum_{k=1}^K ( \phi_{\max,k}-\phi_{\min,k}),
\end{equation}
along with the $\max $-spread
\begin{equation}\label{eq:max-inner-spread}
\Delta (f):= \max_{1\le k \le K} ( \phi_{\max,k}-\phi_{\min,k}),
\end{equation}
where
 \begin{equation}\label{eq:phi-def-spread}
 \phi_{\min,k}:= \inf_{s \in \mathcal{S}_k} f_k (s), \quad \phi_{\max,k}:=\sup_{s \in \mathcal{S}_k} f_k( s).
 \end{equation}
 We define the relative spread with power constraint $P$ as
 \begin{equation}\label{eq:relative-spread}
 \Delta (f\| P):= P \cdot \frac{\bar{\Delta} (f)}{\Delta (f)}.
 \end{equation}

%
%
\section{Main Result}
We are now in a position to state our main theorem on approximation of functions in $ \fr_{\textrm{\textrm{mon}}}$. We use $\operatornorm{\cdot}$ and $\frobeniusnorm{\cdot}$ to denote the operator and Frobenius norm of matrices, respectively.
\begin{theorem}
\label{th:correlated}
Let $f\in  \fr_{\textrm{\textrm{mon}}}$, $M\in \NN$, and the power constraint $P \in \reals_{+} $ be given. Let $\Phi$ be an increment majorant of $f$. Assume the fading and noise are correlated as determined by matrices $\fadingTransformMatrix$ and $\noiseTransformMatrix$. Let $\uncorrelatedApprox \in \reals^{2\numChannelUses\numUsers \times (2\numChannelUses\numUsers + 2\numChannelUses)}$ be a matrix which generates user-uncorrelated fading that approximates $\fadingTransformMatrix$ in the sense that
\[
\operatornorm{(\fadingTransformMatrix+\uncorrelatedApprox)\transpose{(\fadingTransformMatrix-\uncorrelatedApprox)}}
\leq
\fadingApproxError.
\]

Then there exist pre- and post-processing operations such that
\begin{multline}
\label{eq:correlated-statement}
\Probability\left(
  \absolute{\bar{f}-f(s_1, \dots, s_\numUsers)}
  \geq
  \tail
\right)
\\
\leq
\begin{aligned}[t]
&2\exp\left(
-
\frac{\numChannelUses \Phi^{-1}(\tail)^2}{
  16\errorProbFrobeniusNormTerm
  +
  \errorProbDependenceTerm
  +
  4\Phi^{-1}(\tail)
  \errorProbOperatorNormTerm
}
\right)
\\
&+
2\exp\left(
  -
  \frac{\numChannelUses \Phi^{-1}(\tail)^2}{
    256
    \errorProbFrobeniusNormTerm
    +
    32 \Phi^{-1}(\tail)
    \errorProbOperatorNormTerm
  }
\right),
\end{aligned}
\end{multline}
where
\begin{align*}
\errorProbOperatorNormTerm
&=
\left(
  \sqrt{\bar{\Delta}(f)} \operatornorm{\fadingTransformMatrix} + \sqrt{\frac{\Delta(f)}{\powerconstraint}}\operatornorm{\noiseTransformMatrix}
\right)^2
\displaybreak[0] \\
\errorProbFrobeniusNormTerm
&=
\errorProbOperatorNormTerm
\left(
  \sqrt{\frac{\bar{\Delta}(f)}{\numChannelUses}} \frobeniusnorm{\fadingTransformMatrix} + \sqrt{\frac{\Delta(f)}{\powerconstraint\numChannelUses}} \frobeniusnorm{\noiseTransformMatrix}
\right)^2
\displaybreak[0] \\
\errorProbDependenceTerm
&=
\bigg(
  4\sqrt{2\numChannelUses}
  \bar{\Delta}(f)
  \fadingApproxError
  +
  4
  \frac{\Delta(f)}{\sqrt{\powerconstraint\numChannelUses}}
  \frobeniusnorm{\fadingTransformMatrix\transpose{\noiseTransformMatrix}}
\bigg)^2.
\end{align*}
\end{theorem}

In the following, we sketch how the pre- and post-processing schemes of Theorem~\ref{th:correlated} work. For a full formal definition, we refer the reader to the proof in Section~\ref{sec:thproof}. The pre- and post-processing schemes are similar (although not identical) to the ones that have appeared in~\cite{goldenbaum2014channel, kiril}. They are based on the same idea of combining random phase shifts at the transmitter and averaging at the receiver to mitigate the impact of the unknown fading and noise.

We consider fast fading (which can have any sub-Gaussian distribution) and assume that \gls{csi} is available neither at the transmitter nor at the receiver. One example is i.i.d. complex standard normal fading, which has a uniformly distributed phase. In this case, since there is no \gls{csi}, the phase of the received signal cannot carry any useful information. On the other hand, the phase difference between the signals from different transmitters plays a crucial role, since it determines whether the signals constructively or destructively overlap at the receiver. We mitigate the impact of a destructive superposition by applying a random phase shift at the transmitters and averaging the transmission over multiple channel uses. This has the added benefit of averaging out additive noise.

In summary, the pre-processing is performed by applying the following steps at each transmitter $k$ (see Section~\ref{sec:pre-proc}):
\begin{itemize}
 \item Apply the inner function $f_k$ from the nomographic representation (\ref{eq:nomographic-def}).
 \item Shift and rescale to satisfy the power constraint.
 \item Apply a random phase shift $U_k(m)$ that is independent for each channel use $m$.
\end{itemize}
One option for the random phase shift $U_k(m)$ is to draw it uniformly from the complex unit circle, but as we argue in the proofs, it is actually sufficient to draw it uniformly from $\{-1,1\}$.

As described above, the phase of the received signal carries no useful information due to the absence of \gls{csi}. Moreover, to compensate for the phase differences between the transmitters and reduce the influence of additive noise, some form of averaging is required. We therefore perform the following steps (see Section~\ref{sec:post-proc}):
\begin{itemize}
 \item Compute the total energy of the received signal.
 \item Subtract the energy of the additive noise over the receive time slot (this is statistical information and does not require knowledge of the instantaneous noise realizations).
 \item Invert the rescaling and shift that have been applied during pre-processing.
 \item Apply the outer function $F$ from the nomographic representation (\ref{eq:nomographic-def}).
\end{itemize}

We remark that these pre- and post-processing steps, while specific to the fast fading scenario, are unmodified compared to the steps used in~\cite{bjelakovic2019distributed}, where we considered only the case of uncorrelated fading and noise. On the other hand, the error bounds of Theorem~\ref{th:correlated} and the statistical tools used to prove it are different. Moreover, the error bounds depend on the correlation structure of the fading and noise, while the pre- and post-processors do not need this information.

\begin{remark}
If no suitable user-uncorrelated approximation for $\fadingTransformMatrix$ is available, we can always choose $\uncorrelatedApprox := 0$, which results in
$
 \fadingApproxError = \operatornorm{\fadingTransformMatrix}^2.
$
\end{remark}

\begin{remark}
In order to gain more freedom for optimizing the bound for a given correlation structure, it is possible to replace $\uncorrelatedApprox$ in Theorem~\ref{th:correlated} with $\uncorrelatedApprox \unitaryFadingApprox$, where $\unitaryFadingApprox \in \reals^{(2\numChannelUses\numUsers + 2\numChannelUses) \times (2\numChannelUses\numUsers + 2\numChannelUses)}$ is a unitary matrix. This requires only minor adaptations in the proof of the theorem.
\end{remark}

\begin{cor}
\label{cor:communication-cost}
For the approximation communication cost, we have
\begin{align}\label{eq:commcost}
M(f,\varepsilon,\delta)
\leq
\frac{\log 4 - \log \delta}{\Phi^{-1}(\varepsilon)^2}
\Gamma,
\end{align}
where
\begin{multline*}
\Gamma
:=
\max\big(
  16\errorProbFrobeniusNormTerm
  +
  \errorProbDependenceTerm
  +
  4\Phi^{-1}(\tail)
  \errorProbOperatorNormTerm
  ,
  256
  \errorProbFrobeniusNormTerm
  +
  32 \Phi^{-1}(\tail)
  \errorProbOperatorNormTerm
\big).
\end{multline*}
\end{cor}
\begin{proof}
We upper bound (\ref{eq:correlated-statement}) as
\begin{align*}
 \mathbb{P} (  |    \bar{f} (s^K)- f(s^K)       |\ge \eps       )
 \le  
4\exp\left(
-\frac{\numChannelUses \Phi^{-1}(\tail)^2}{\Gamma}
\right),
\end{align*}
and solve the expression for $M$ concluding the proof.
\end{proof}
\begin{remark}
If $F$ is $C$-Lipschitz continuous, we can replace $\Phi^{-1}(\varepsilon)$ in (\ref{eq:commcost}) and the expression for $\Gamma$ with $\varepsilon/C$.
\end{remark}

\subsection{Special cases of Theorem~\ref{th:correlated}}
\label{sec:specialcases}
In this subsection we discuss the bound of Theorem~\ref{th:correlated} and illustrate it with two examples. In order to be a useful bound, (\ref{eq:correlated-statement}) should approach $0$ as $\numChannelUses \rightarrow \infty$. Clearly, it does so exponentially whenever $\errorProbOperatorNormTerm$, $\errorProbFrobeniusNormTerm$ and $\errorProbDependenceTerm$ are bounded for $\numChannelUses \rightarrow \infty$.

For $\errorProbDependenceTerm$, we observe that $\fadingTransformMatrix\transpose{\noiseTransformMatrix}$ is $0$ if the fading is independent of the additive noise (which we usually expect to be the case in practically relevant scenarios) and that $\fadingApproxError$ is $0$ in the case of user-uncorrelated fading in the sense of Definition~\ref{def:useruncorrelated}, which means that the fading of any one user is independent of the fading of the other users (arbitrary correlations in the time domain are still allowed). Since $\errorProbDependenceTerm$ grows proportionally with $\numChannelUses \fadingApproxError^2$, we can see that our bound is not useful for the case of strong correlations between users. Therefore, in the presence of user-correlated fading, the usefulness of the bound depends on the scaling behavior of $\numChannelUses \fadingApproxError^2$. When this term exhibits sublinear growth, the error bound of Theorem~\ref{th:correlated} approaches $0$ as $\numChannelUses \rightarrow 0$, and if it is additionally upper bounded, the error bound does so exponentially. In this sense, the bound is robust to small deviations from the assumption of user-uncorrelatedness. However, it is important to note that even the user-uncorrelated case covers relevant cases of correlation in the time domain such as the block fading channel. In the expression of $\errorProbFrobeniusNormTerm$, we can see that the Frobenius norm of both $\fadingTransformMatrix$ and $\noiseTransformMatrix$ should not grow faster than $\sqrt{\numChannelUses}$ and finally, in the expression of $\errorProbOperatorNormTerm$, we see that the operator norms of $\fadingTransformMatrix$ and $\noiseTransformMatrix$ should not grow with $\numChannelUses$. We illustrate that this is the case in scenarios of interest with the following two examples.

\begin{cor}
\label{cor:uncorrelated}
In the setting of Theorem~\ref{th:correlated} with uncorrelated fading and noise, i.e.,
\begin{align}
\label{eq:uncorrelated-hypothesis}
\fadingTransformMatrix
:=
\begin{pmatrix}
\sigma_F \identityMatrix_{2\numChannelUses\numUsers} & 0
\end{pmatrix}
,~~
\noiseTransformMatrix
:=
\begin{pmatrix}
0 & \sigma_N \identityMatrix_{2\numChannelUses}
\end{pmatrix},
\end{align} 
where $\identityMatrix_n$ denotes the $n \times n$ identity matrix,
we have
\begin{multline}
\label{eq:uncorrelated-statement}
\Probability\left(
  \absolute{\bar{f}-f(s_1, \dots, s_\numUsers)}
  \geq
  \tail
\right)
\\
\leq
\begin{aligned}[t]
&2\exp\left(
-
\frac{\numChannelUses \Phi^{-1}(\tail)^2}{
  16\errorProbFrobeniusNormTerm'
  +
  4\Phi^{-1}(\tail)
  \errorProbOperatorNormTerm'
}
\right)
\\
&+
2\exp\left(
  -
  \frac{\numChannelUses \Phi^{-1}(\tail)^2}{
    256
    \errorProbFrobeniusNormTerm'
    +
    32 \Phi^{-1}(\tail)
    \errorProbOperatorNormTerm'
  }
\right),
\end{aligned}
\end{multline}
where
\begin{align*}
\errorProbOperatorNormTerm'
&=
\left(
  \sqrt{\bar{\Delta}(f)} \sigma_F + \sqrt{\frac{\Delta(f)}{\powerconstraint}}\sigma_N
\right)^2
\displaybreak[0] \\
\errorProbFrobeniusNormTerm'
&=
\errorProbOperatorNormTerm'
\left(
  \sqrt{2\numUsers\bar{\Delta}(f)} \sigma_F
  +
  \sqrt{\frac{2\Delta(f)}{\powerconstraint}}\sigma_N
\right)^2.
\end{align*}
\end{cor}
\begin{proof}
Note that
$
\fadingTransformMatrix \transpose{\noiseTransformMatrix} = 0
$,
$
\operatornorm{\fadingTransformMatrix} = \sigma_F
$,
$
\operatornorm{\noiseTransformMatrix} = \sigma_N
$,
$
\frobeniusnorm{\fadingTransformMatrix} = \sqrt{2\numChannelUses\numUsers} \sigma_F
$ and
$
\frobeniusnorm{\noiseTransformMatrix} = \sqrt{2\numChannelUses} \sigma_N
$;
pick
$
\uncorrelatedApprox := \fadingTransformMatrix
$
and substitute this into (\ref{eq:correlated-statement}).
\end{proof}

\begin{cor}
\label{cor:blockfading}
In the setting of Theorem~\ref{th:correlated} where each user has a block fading channel with block length $\fadingBlockLength$, i.e.,
\newcommand\bigzero{\makebox(0,0){\text{\huge0}}}
\begin{align*}
\fadingTransformMatrix
&:=
\sigma_F
\left(
\begin{array}{cccc:c}
\identityMatrix_{2\numUsers} & & & & \\
\vdots &  &  & &\\
\identityMatrix_{2\numUsers} & & & \bigzero & \\
 & \identityMatrix_{2\numUsers} & & \\
 & \vdots & & &\\
 & \identityMatrix_{2\numUsers} & & & ~~\bigzero \\
 & & \ddots & & \\
 \bigzero & & & \identityMatrix_{2\numUsers} & \\
 & & & \vdots & \\
 & & & \identityMatrix_{2\numUsers} &
\end{array}
\right)
\hspace{-1em}
\begin{tabular}{l}
$\left.\lefteqn{\phantom{\begin{array}{c} \identityMatrix_{2\numUsers} \\ \vdots\\ \identityMatrix_{2\numUsers} \ \end{array}}}\right\}\fadingBlockLength$\\
$\left.\lefteqn{\phantom{\begin{array}{c} \identityMatrix_{2\numUsers} \\ \vdots\\ \identityMatrix_{2\numUsers} \ \end{array}}}\right\}\fadingBlockLength$\\
\lefteqn{\phantom{\begin{array}{c} \ddots \ \end{array}}} \\
$\left.\lefteqn{\phantom{\begin{array}{c} \identityMatrix_{2\numUsers} \\ \vdots\\ \identityMatrix_{2\numUsers} \ \end{array}}}\right\}\fadingBlockLength$
\end{tabular}
\\
\noiseTransformMatrix
&:=
\begin{pmatrix}
0 & \sigma_N \identityMatrix_{2\numChannelUses},
\end{pmatrix}
\end{align*}
we have a bound of the form (\ref{eq:uncorrelated-statement}), where
\begin{align*}
\errorProbOperatorNormTerm'
&=
\left(
  \sqrt{\bar{\Delta}(f) \fadingBlockLength} \sigma_F + \sqrt{\frac{\Delta(f)}{\powerconstraint}}\sigma_N
\right)^2
\\
\errorProbFrobeniusNormTerm'
&=
\errorProbOperatorNormTerm'
\left(
  \sqrt{2\numUsers\bar{\Delta}(f)} \sigma_F
  +
  \sqrt{\frac{2\Delta(f)}{\powerconstraint}}\sigma_N
\right)^2.
\end{align*}
\end{cor}
\begin{proof}
Note that
$
\fadingTransformMatrix \transpose{\noiseTransformMatrix} = 0
$,
$
\operatornorm{\fadingTransformMatrix} = \sigma_F \sqrt{\fadingBlockLength}
$,
$
\operatornorm{\noiseTransformMatrix} = \sigma_N
$,
$
\frobeniusnorm{\fadingTransformMatrix} = \sqrt{2\numChannelUses\numUsers} \sigma_F
$ and
$
\frobeniusnorm{\noiseTransformMatrix} = \sqrt{2\numChannelUses} \sigma_N
$;
pick
$
\uncorrelatedApprox := \fadingTransformMatrix
$
and substitute this into (\ref{eq:correlated-statement}).
\end{proof}

\subsection{Sharpness of the Bound in Theorem~\ref{th:correlated}}
\label{sec:sharpness}
We do not expect the bound (\ref{eq:correlated-statement}) to be sharp in the sense that there are non-trivial examples in which it holds with equality. This, we believe, is in part a price that we pay for using a very general system model, but it is also due to the underlying tools from high-dimensional statistics that we employ. A further sharpening of this bound could be an interesting question for future research, but it would hinge on optimizing the bounds of some of the basic results that we use (such as Lemma~\ref{lemma:hw-offdiagonal}). In some special cases (such as uncorrelated Gaussian noise and fading) it is not hard, however, to compute exact bounds, as can for instance be seen in (\ref{eq:sharpening-fY}) below. In the sequel, we argue that in a sense that will be made precise, the bound (\ref{eq:correlated-statement}) is sharp ``up to absolute constants''. The example case for which we show that the bound holds with equality up to constants is uncorrelated Gaussian fading and noise so that the bound specializes to (\ref{eq:uncorrelated-statement}). For the purpose of this section, we focus on the behavior with varying $\numChannelUses$ and $\tail$, while we consider everything else constant system parameters.

\begin{theorem}
\label{th:sharpness}
In the case of uncorrelated Gaussian fading and noise; i.e., (\ref{eq:uncorrelated-hypothesis}) is satisfied and the entries of $\randomBaseVector$ are i.i.d. standard Gaussian, there are constants $\absoluteConstantOne$ and $\absoluteConstantTwo$ such that the estimate $\bar{f}$ obtained by the pre- and post-processing schemes described in Sections \ref{sec:pre-proc} and \ref{sec:post-proc} satisfies
\begin{multline}
\label{eq:sharpness}
\Probability\left(
  \absolute{\bar{f}-f(s_1, \dots, s_\numUsers)}
  \geq
  \tail
\right)
\\
\geq
\absoluteConstantOne
\exp\left(
  -
  \absoluteConstantTwo
  \numChannelUses \min(\Phi^{-1}(\tail), \Phi^{-1}(\tail)^2)
\right)
\end{multline}
for suitable choices of $F$ and $\Phi$ such as $F=\Phi=\identityFunction$ (the identity function).
\end{theorem}
Note that the upper tail bound (\ref{eq:uncorrelated-statement}) also has the same form for suitably chosen $\absoluteConstantOne$ and $\absoluteConstantTwo$, so that we can conclude that it is sharp up to the values of these constants.
\begin{proof}
The proof is relatively straightforward, so we only sketch it. Under the assumptions made in Theorem~\ref{th:sharpness}, we readily compute from the equations in Sections \ref{sec:pre-proc} and \ref{sec:post-proc}
\[
\frac{\vectorNorm{\rxVector}^2}
     {\sigma_F^2 \sum_{k=1}^K \transmitPower_k + \sigma_N^2}
=
\vectorNorm{\randomBaseVector}^2.
\]
Since the entries of $\randomBaseVector$ are i.i.d. standard Gaussian and the vector has $2\numChannelUses$ entries, $\vectorNorm{\randomBaseVector}^2$ clearly follows a chi-square distribution with $2\numChannelUses$ degrees of freedom. We therefore have in parallel to (\ref{eq:correlated-fY})
\begin{multline}
\label{eq:sharpening-fY}
\Probability\left(
  \absolute{\bar{f}-f(s_1, \dots, s_\numUsers)}
  \geq
  \tail
\right)
\\
\leq
\Probability\left(
  \absolute{\vectorNorm{\randomBaseVector}^2 - \Expectation \vectorNorm{\randomBaseVector}^2}
  \geq
  \frac{2\powerconstraint\numChannelUses\Phi^{-1}(\tail)}{\Delta(f)({\sigma_F^2 \sum_{k=1}^K \transmitPower_k + \sigma_N^2})}
\right).
\end{multline}
The bound here is sharp in the sense that it holds with equality in case $F = \Phi = \identityFunction$. We can now use \cite[Corollary 3]{zhang2018non} to conclude that in this case, (\ref{eq:sharpness}) holds for suitable $\absoluteConstantOne$ and $\absoluteConstantTwo$.
\end{proof}

\section{Applications to Vertical Federated Learning}
\label{sec:vfl}
In this section, we give examples of how the results of this paper can be applied to \gls{ml} problems. We focus on a case that is particularly simple on the communication side (since only one weighted sum is \gls{ota} computed) but which we expect to gain significant relevance in practical systems. In this section, we first give an overview of literature about distributed and \gls{ota} \gls{ml}, then we give two detailed examples of how an \gls{ota} computation of weighted sums can be applied to determining labels in \gls{ml} problems. In the last part of the section, we present simulation results for one of these application examples.

Prior works~\cite{amiri2020machine,ahn2019wireless,mcmahan2017communication,konecny2016federated,zhu2019broadband,amiri2020federated,yang2020federated,sery2020analog,sun2019energy,seif2020wireless,amiri2019collaborative} on \gls{fl} using \gls{ota} computation have focused on distributed training in the case that the agents in the system observe the same feature space, but their training data is partitioned in the sample domain. This setting is sometimes referred to as \emph{Horizontal \gls{fl}}~\cite{yang2019federated}. Stochastic Gradient Descent, which is a very common training algorithm in Horizontal FL, has sum form and is therefore contained in $\fr_{\textrm{mon}}$. For this reason, we believe that it is a promising direction for future research to investigate whether Theorem~\ref{th:correlated} could be used to expand, e.g., the distributed \gls{ota} computation based training procedure described in~\cite{amiri2020machine} to accommodate the sub-Gaussian fast-fading channel model with no instantaneous \gls{csi} and possible correlations in fading and noise that we consider in the present work.

In this paper, on the other hand, we show applications to a different \gls{fl} setting, which is called \emph{\gls{vfl}}~\cite{yang2019federated}. In this setting, each agent observes only a projection of the features, but the same samples (in the training phase, along with the full labels). Horizontal schemes carry out the training phase in a distributed manner, but in the prediction phase, they evaluate the classifier or regressor at a central point. Vertical schemes, on the other hand, can also perform the labeling in a decentralized way. In this section, we describe a few examples of how our \gls{dfa} scheme can be leveraged to vastly increase the efficiency of the prediction phase both in terms of time and bandwidth resources (i.e., in our model, channel uses) and in terms of energy resources expended. For the training phase, we will either assume centralized offline training or use more communication-efficient decentralized methods that do not, however, leverage any form of \gls{ota} computation. Developing distributed training algorithms for \gls{vfl} which can leverage the full power of \gls{ota} computation remains open as an interesting future research direction. However, we argue that in many cases of interest the communication cost incurred in the prediction phase can dominate that incurred in training and it is therefore worthwhile to focus on the prediction. This can, e.g., be the case when the training can be conducted offline and the models do not or only infrequently have to be re-trained; or when the number of training samples is small, but the number of features observed in the system is large and prediction tasks have to be carried out very frequently.

Contrary to~\cite{yang2019federated}, where the main focus in \gls{vfl} is on providing privacy and security guarantees, in this work we focus on the communication efficiency of such schemes under the use of \gls{ota} computation. Since the \gls{ota}-computed predictors as well as the distributed training procedures we describe do not aggregate the observed features in a central point, it is reasonable to expect that these methods have good inherent privacy properties, and for some of the envisioned applications, such as, e.g., e-health, it is an important open question for future research how these privacy guarantees can be formalized and perhaps strengthened in the context of \gls{ota}-computed \gls{ml} predictors.

An application of \gls{ota} computation to the prediction phase of \gls{vfl} has, to the best of our knowledge, first appeared in~\cite{kiril}, where the authors proposed an \gls{ota}-computed classifier based on a \gls{svm} with a linear kernel to approach the classification problem of anomaly detection. Moreover, the research area of \emph{Type-Based Multiple-Access}~\cite{mergen2006type,mergen2007asymptotic} is concerned with solving problems very similar to those approached by \gls{ota} \gls{vfl} such as anomaly detection in extremely large networks. An important difference is that instead of transmitting analog values directly, this approach relies on a prior quantization step and then exploits the fact that the number of quantization levels usually does not grow with the number of transmitters in the system. Furthermore, this approach uses statistical methods and knowledge about the involved probability distributions, while the \gls{vfl} approach that we use here is based on the \gls{ml} paradigm and hence does not require a priori knowledge of the underlying distributions.

In the following, we expand upon the idea of \gls{ota} \gls{vfl} by showing in Section~\ref{sec:vfl-additive-svm} how the \gls{svm} approach can be generalized to the use of additive kernel \glspl{svm} and applied to regression and classification problems, and in Section~\ref{sec:vfl-classification} how classifiers that do not necessarily have to rely on \glspl{svm} at all can be constructed to solve binary classification problems. In Section~\ref{sec:vfl-classification-numerical}, we present simulation results of two classification schemes constructed as described in Section~\ref{sec:vfl-classification} and compare them to two baseline approaches.

In both Subsection~\ref{sec:vfl-additive-svm} and~\ref{sec:vfl-classification}, we construct an \gls{ml} regressor or predictor that has the form of a weighted sum, because such a function can be straightforwardly computed \gls{ota} using the \gls{dfa} scheme described in Section~\ref{sec:thproof}. If the loss is Lipschitz-continuous, it can play the role of the function $F$ in Definition~\ref{def:Fmon} so that Theorem~\ref{th:correlated} provides a tail bound on the additional \gls{ml} loss that the \gls{ota}-computed classifiers can incur in addition to the loss they would have in case of noiseless communication. The detailed technical statements and proofs of these facts can be found in Corollaries~\ref{cor:mlapplication} and~\ref{cor:boostapplication}. We also give some examples of applicable Lipschitz-continuous losses for regression and classification that are commonly used in practice in Subsection~\ref{sec:vfl-additive-svm}.

\subsection{SVMs with Additive Kernels for Regression and Classification}
\label{sec:vfl-additive-svm}
{
\newcommand{\mlInputAlphabet}{{\mathcal{X}}}
\newcommand{\mlLabelAlphabet}{{\mathcal{Y}}}
\newcommand{\mlInputAlphabetElement}{{x}}
\newcommand{\mlLabelAlphabetElement}{{y}}
\newcommand{\mlEstimator}{{f}}
\newcommand{\mlLoss}{{L}}
\newcommand{\mlDistribution}{{\mathcal{P}}}
\newcommand{\mlInputRV}{{X}}
\newcommand{\mlOutputRV}{{Y}}
\newcommand{\mlRisk}[2]{{\mathcal{R}_{{#1},{#2}}}}
\newcommand{\mlEstimatorOutput}{{t}}
\newcommand{\mlKernel}{{\kappa}}
\newcommand{\RKHS}{{\mathcal{H}}}
\newcommand{\mlTrainingSampleNum}{{N}}
\newcommand{\mlIndexTrainingSample}{{n}}
\newcommand{\mlEstimatorCoefficient}{{\alpha}}
\newcommand{\partfuncminvalue}[1]{{\phi_{\min,{#1}}}}
\newcommand{\partfuncmaxvalue}[1]{{\phi_{\max,{#1}}}}
\newcommand{\errorconstone}{{\varepsilon}}
\newcommand{\errorconsttwo}{{\delta}}
\newcommand{\mlLossLipschitz}{{B}}

In this subsection, we give an example of additive, and therefore \gls{ota} computable, \gls{svm} regressors. First, we briefly sketch the setting as in~\cite{steinwart}. We consider a feature alphabet $\mlInputAlphabet$, a label alphabet $\mlLabelAlphabet \subseteq \reals$ and a probability distribution $\mlDistribution$ on $\mlInputAlphabet \times \mlLabelAlphabet$ which is in general unknown. A statistical inference problem is characterized by the feature alphabet, the label alphabet and a loss function $\mlLoss: \mlInputAlphabet \times \mlLabelAlphabet \times \reals \rightarrow [0, \infty)$. The objective is, given training samples drawn i.i.d. from $\mlDistribution$, to find an estimator function $\mlEstimator: \mlInputAlphabet \rightarrow \reals$ such that the risk $\mlRisk{\mlLoss}{\mlDistribution} := \Expectation_\mlDistribution \mlLoss(\mlInputRV, \mlOutputRV, \mlEstimator(\mlInputRV))$ is as small as possible. In order for the risk to exist, we must impose suitable measurability conditions on $\mlLoss$ and $\mlEstimator$. In this paper, we deal with Lipschitz-continuous losses. We say that the loss $\mlLoss$ is $\mlLossLipschitz$-Lipschitz-continuous if $\mlLoss(\mlInputAlphabetElement, \mlLabelAlphabetElement, \cdot)$ is Lipschitz-continuous for all $\mlInputAlphabetElement \in \mlInputAlphabet$ and $\mlLabelAlphabetElement \in \mlLabelAlphabet$ with a Lipschitz constant uniformly bounded by $\mlLossLipschitz$. Lipschitz-continuity of a loss function is a property that is also often needed in other contexts. Fortunately, many loss functions of practical interest possess this property. For instance, the absolute distance loss, the logistic loss, the Huber loss and the $\varepsilon$-insensitive loss, all of which are commonly used in regression problems~\cite[Section 2.4]{steinwart}, are Lipschitz-continuous. Even in scenarios in which the naturally arising loss is not Lipschitz-continuous, for the purpose of designing the \gls{ml} model, it is often replaced with a Lipschitz-continuous alternative. For instance, in binary classification, we have $\mlLabelAlphabet = \{-1,1\}$ and the loss function is given by
\[
(\mlInputAlphabetElement,\mlLabelAlphabetElement,\mlEstimatorOutput) \mapsto
\begin{cases}
  0, &\sign(\mlLabelAlphabetElement) = \sign(\mlEstimatorOutput) \\
  1, &\text{otherwise.}
\end{cases}
\]
This loss is not even continuous, which makes it hard to deal with. So for the purpose of designing the \gls{ml} model, it is commonly replaced with the Lipschitz-continuous hinge loss or logistic loss~\cite[Section 2.3]{steinwart}.

Here, we consider the case in which the features are $\numUsers$-tuples and the \gls{svm} can be trained in a centralized fashion. The actual predictions, however, are performed in a distributed setting; i.e., there are $\numUsers$ users each of which observes only one component of the features. The objective is to make an estimate of the label available at the receiver while using as little communication resources as possible.

To this end, we consider the case of additive models which is described in~\cite[Section 3.1]{christmann2012consistency}. We have $\mlInputAlphabet = \mlInputAlphabet_1 \times \dots \times \mlInputAlphabet_\numUsers$ and a kernel $\mlKernel_\indexUsers: \mlInputAlphabet_\indexUsers \times \mlInputAlphabet_\indexUsers \rightarrow \reals$ with an associated reproducing kernel Hilbert space $\RKHS_\indexUsers$ of functions mapping from $\mlInputAlphabet_\indexUsers$ to $\reals$ for each $\indexUsers \in \{1,\dots,\numUsers\}$. Then by~\cite[Theorem 2]{christmann2012consistency}
\begin{multline}\label{eq:addkernel}
\mlKernel: \mlInputAlphabet \times \mlInputAlphabet \rightarrow \reals,~
((\mlInputAlphabetElement_1,\dots,\mlInputAlphabetElement_\numUsers), (\mlInputAlphabetElement'_1,\dots,\mlInputAlphabetElement'_\numUsers))
\mapsto \\
\mlKernel_1(\mlInputAlphabetElement_1, \mlInputAlphabetElement'_1) + \dots + \mlKernel_\numUsers(\mlInputAlphabetElement_\numUsers, \mlInputAlphabetElement'_\numUsers)
\end{multline}
is a kernel and the associated reproducing kernel Hilbert space is
\begin{align}\label{eq:addkernel-rkhs}
\RKHS := \{\mlEstimator_1 + \dots + \mlEstimator_\numUsers: \mlEstimator_1 \in \RKHS_1, \dots, \mlEstimator_\numUsers \in \RKHS_\numUsers\}.
\end{align}
So this model is appropriate whenever the function to be approximated is expected to have an additive structure. We know~\cite[Theorem 5.5]{steinwart} that an \gls{svm} estimator has the form
\begin{align}\label{eq:svmestimator}
\mlEstimator(\mlInputAlphabetElement) = \sum_{\mlIndexTrainingSample=1}^\mlTrainingSampleNum \mlEstimatorCoefficient_\mlIndexTrainingSample \mlKernel(\mlInputAlphabetElement, \mlInputAlphabetElement^\mlIndexTrainingSample),
\end{align}
where $\mlEstimatorCoefficient_1, \dots, \mlEstimatorCoefficient_\mlTrainingSampleNum \in \reals$ and $\mlInputAlphabetElement^1, \dots \mlInputAlphabetElement^\mlTrainingSampleNum \in \mlInputAlphabet$. In our additive model, this is
\begin{align}\label{eq:addestimator}
\mlEstimator(\mlInputAlphabetElement_1, \dots, \mlInputAlphabetElement_\indexUsers)
=
\sum_{\indexUsers=1}^\numUsers \mlEstimator_\indexUsers(\mlInputAlphabetElement_\indexUsers),
\end{align}
where for each $\indexUsers$,
\begin{align}\label{eq:addestimator-detail}
\mlEstimator_\indexUsers(\mlInputAlphabetElement_\indexUsers) =
\sum_{\mlIndexTrainingSample=1}^\mlTrainingSampleNum \mlEstimatorCoefficient_\mlIndexTrainingSample \mlKernel_\indexUsers(\mlInputAlphabetElement_\indexUsers, \mlInputAlphabetElement_\indexUsers^\mlIndexTrainingSample).
\end{align}

We now state a result for the distributed approximation of the estimator of such an additive model as an immediate consequence of Theorem~\ref{th:correlated}.

\begin{cor}
\label{cor:mlapplication}
Consider an additive \gls{ml} model, i.e., we have an estimator of the form (\ref{eq:addestimator}), and assume that $\mlLoss$ is a $\mlLossLipschitz$-Lipschitz-continuous loss. Suppose further that all the $\mlEstimator_\numUsers$ have bounded range such that the quantities $\bar{\Delta}(f)$ and $\Delta(f)$ as defined in (\ref{eq:total-inner-spread}) and (\ref{eq:max-inner-spread}) exist and are finite. Let $\errorconstone, \errorconsttwo > 0$ and $M \geq M(\mlEstimator,\errorconstone, \errorconsttwo)$ as defined in (\ref{eq:commcost}), where $\Phi := \identityFunction$ and thus $\Phi^{-1}(\varepsilon) = \varepsilon$. Then, given $\mlInputAlphabetElement^\numUsers = (\mlInputAlphabetElement_1, \dots, \mlInputAlphabetElement_\numUsers)$ drawn from an arbitrary distribution\footnote{\label{footnote:correlated}Arbitrary distribution means in particular that the components can be arbitrarily correlated.}\footnote{\label{footnote:stronger-result}Note that Theorem~\ref{th:correlated} actually provides for a stronger result, since it allows arbitrary deterministic values, which implies the applicability to arbitrarily distributed random variables through the law of total probability.} at the transmitters and any $\mlLabelAlphabetElement \in \mlLabelAlphabet$, through $M$ uses of the channel (\ref{eq:channel-model}), the receiver can obtain an estimate $\bar{\mlEstimator}$ of $\mlEstimator(\mlInputAlphabetElement^\numUsers)$ satisfying
\begin{align}
\label{eq:mlapplication-cor}
\Probability(\absolute{\mlLoss(\mlInputAlphabetElement^\numUsers,\mlLabelAlphabetElement,\bar{\mlEstimator}) - \mlLoss(\mlInputAlphabetElement^\numUsers,\mlLabelAlphabetElement,\mlEstimator(\mlInputAlphabetElement^\numUsers))} \geq \mlLossLipschitz \errorconstone)
\leq \errorconsttwo.
\end{align}
\end{cor}

\begin{proof}
The Lipschitz continuity of $\mlLoss$ yields
\begin{multline*}
\Probability(\absolute{\mlLoss(\mlInputAlphabetElement^\numUsers,\mlLabelAlphabetElement,\bar{\mlEstimator}) - \mlLoss(\mlInputAlphabetElement^\numUsers,\mlLabelAlphabetElement,\mlEstimator(\mlInputAlphabetElement^\numUsers))} \geq \mlLossLipschitz \errorconstone)
\\
\leq
\Probability(\absolute{\bar{\mlEstimator} - \mlEstimator(\mlInputAlphabetElement^\numUsers)} \geq \errorconstone),
\end{multline*}
from which (\ref{eq:mlapplication-cor}) follows by the definition of $M(\mlEstimator,\errorconstone, \errorconsttwo)$.
\end{proof}

While~\cite{christmann2012consistency} provides some examples of applications of \glspl{svm} with additive kernels to regression problems, the example for anomaly detection described in~\cite{kiril} can be recovered as a special case of the framework described in this subsection, where the employed \glspl{svm} have linear kernels.

We conclude this subsection with a brief discussion of the feasibility of the condition that $\mlEstimator_1, \dots, \mlEstimator_\numUsers$ have bounded ranges in the case of the additive \gls{svm} model discussed above. The coefficients $\mlEstimatorCoefficient_1, \dots, \mlEstimatorCoefficient_\mlTrainingSampleNum$ are a result of the training step and can therefore be considered constant, so all we need is that the ranges of $\mlKernel_1, \dots, \mlKernel_\numUsers$ are bounded. This heavily depends on $\mlInputAlphabet_1, \dots, \mlInputAlphabet_\numUsers$ and the choices of the kernels, but we remark that the boundedness criterion is satisfied in many cases of interest. The range of Gaussian kernels is always a subset of $(0,1]$, and while other frequent choices such as exponential, polynomial and linear kernels can have arbitrarily large ranges, they are nonetheless continuous which means that as long as the input alphabets are compact topological spaces (e.g., closed hyperrectangles or balls), the ranges are also compact, and therefore bounded.
}

\subsection{Model-Agnostic Approach to \gls{ota}-Computed Classifiers}
\label{sec:vfl-classification}
{
\newcommand{\mlInputAlphabet}{{\mathcal{X}}}
\newcommand{\mlLabelAlphabet}{{\mathcal{Y}}}
\newcommand{\mlInputAlphabetElement}{{x}}
\newcommand{\mlLabelAlphabetElement}{{y}}
\newcommand{\mlEstimator}{{f}}
\newcommand{\mlLoss}{{L}}
\newcommand{\mlDistribution}{{\mathcal{P}}}
\newcommand{\mlInputRV}{{X}}
\newcommand{\mlOutputRV}{{Y}}
\newcommand{\mlRisk}[2]{{\mathcal{R}_{{#1},{#2}}}}
\newcommand{\mlEstimatorOutput}{{t}}
\newcommand{\mlKernel}{{\kappa}}
\newcommand{\RKHS}{{\mathcal{H}}}
\newcommand{\mlEstimatorCoefficient}{{\alpha}}
\newcommand{\partfuncminvalue}[1]{{\phi_{\min,{#1}}}}
\newcommand{\partfuncmaxvalue}[1]{{\phi_{\max,{#1}}}}
\newcommand{\errorconstone}{{\varepsilon}}
\newcommand{\errorconsttwo}{{\delta}}
\newcommand{\mlLossLipschitz}{{B}}
\newcommand{\mlNumTrainingSamples}{{S}}
\newcommand{\mlIndexTrainingSamples}{{s}}
\newcommand{\mlBaseClassifier}[1]{{g_{#1}}}
\newcommand{\mlBoostClassifier}{{g}}
\newcommand{\mlBoostWeight}{{\alpha}}
\newcommand{\adaBoostNumIterations}{{T}}
\newcommand{\adaBoostIndexIterations}{{t}}
\newcommand{\adaBoostUserIndex}{{h}}
\newcommand{\adaBoostDist}[1]{{D_{#1}}}
\newcommand{\adaBoostBaseError}[1]{{\epsilon_{#1}}}
\newcommand{\adaBoostNormalization}[1]{{Z_{#1}}}

In this subsection, we focus on classification problems. The general approach is model-agnostic in the sense that arbitrary and even different \gls{ml} models can be used in the distributed agents, but we have decentralized classifiers with a low computational burden in mind, as is exemplified in the numerical simulations discussed in the following subsection.

We consider a feature alphabet $\mlInputAlphabet = \mlInputAlphabet_1 \times \dots \times \mlInputAlphabet_\numUsers$ and a label alphabet $\mlLabelAlphabet = \{-1, 1\}$ as well as an unknown, but fixed probability distribution $\mlDistribution$ on $\mlInputAlphabet \times \mlLabelAlphabet$. In the training phase, each user $\indexUsers$ observes $\mlNumTrainingSamples$ training samples
\[
T_\indexUsers
=
\left(
  \left(\mlInputAlphabetElement_\indexUsers^{(1)}, \mlLabelAlphabetElement^{(1)}\right),
  \dots,
  \left(\mlInputAlphabetElement_\indexUsers^{(\mlNumTrainingSamples)}, \mlLabelAlphabetElement^{(\mlNumTrainingSamples)}\right)
\right),
\]
where for all $\indexUsers, \mlIndexTrainingSamples$, we have $\mlInputAlphabetElement_\indexUsers^{(\mlIndexTrainingSamples)} \in \mlInputAlphabet_\indexUsers$, $\mlLabelAlphabetElement^{(\mlIndexTrainingSamples)} \in \mlLabelAlphabet$ and $(\mlInputAlphabetElement_1^{(\mlIndexTrainingSamples)}, \dots, \mlInputAlphabetElement_\numUsers^{(\mlIndexTrainingSamples)}, \mlLabelAlphabetElement^{(\mlIndexTrainingSamples)})$ is drawn according to $\mlDistribution$.

Each user $\indexUsers$ can train its own model based on $T_\indexUsers$ which is distributed according to the marginal of $\mlDistribution$ with respect to $\mlInputAlphabet_\indexUsers \times \mlLabelAlphabet$. We propose to use a slight variation of the well-known boosting technique and define a classifier
\begin{align}
\label{eq:boostclassifier}
\mlBoostClassifier := \sum_{\indexUsers=1}^{\numUsers}\mlBoostWeight_\indexUsers \mlBaseClassifier{\indexUsers},
\end{align}
where $\mlBaseClassifier{\indexUsers}$ is the base classifier locally trained at user $\indexUsers$ and $\mlBoostWeight_\indexUsers$ is a nonnegative weight. As an immediate corollary to Theorem~\ref{th:correlated} parallel to Corollary~\ref{cor:mlapplication}, $\mlBoostClassifier$ can be approximated at a central node in a distributed manner.
\begin{cor}
\label{cor:boostapplication}
Assume that $\mlLoss$ is a $\mlLossLipschitz$-Lipschitz-continuous loss. Let $\errorconstone, \errorconsttwo > 0$ and $M \geq M(\mlBoostClassifier,\errorconstone, \errorconsttwo)$ as defined in (\ref{eq:commcost}), where $\Phi^{-1}(\varepsilon) = \varepsilon$, noting that
\begin{align*}
\bar{\Delta}(\mlBoostClassifier) = 2 \sum_{\indexUsers=1}^\numUsers \mlBoostWeight_\indexUsers
,~~
\Delta(\mlBoostClassifier) = 2 \max_{\indexUsers=1}^\numUsers \mlBoostWeight_\indexUsers.
\end{align*}
Then, given any $\mlInputAlphabetElement^\numUsers = (\mlInputAlphabetElement_1, \dots, \mlInputAlphabetElement_\numUsers)$ drawn from an arbitrary\footnoteref{footnote:correlated}\footnoteref{footnote:stronger-result} distribution at the transmitters and any $\mlLabelAlphabetElement \in \mlLabelAlphabet$, through $M$ uses of the channel (\ref{eq:channel-model}), the receiver can obtain an estimate $\bar{\mlBoostClassifier}$ of $\mlBoostClassifier(\mlInputAlphabetElement^\numUsers)$ satisfying
\begin{align}
\label{eq:boostapplication-cor}
\Probability(\absolute{\mlLoss(\mlInputAlphabetElement^\numUsers,\mlLabelAlphabetElement,\bar{\mlBoostClassifier}) - \mlLoss(\mlInputAlphabetElement^\numUsers,\mlLabelAlphabetElement,\mlBoostClassifier(\mlInputAlphabetElement^\numUsers))} \geq \mlLossLipschitz \errorconstone)
\leq \errorconsttwo.
\end{align}
\end{cor}
The proof is the same as for Corollary~\ref{cor:mlapplication}.

It is important to remark here that the predictor $\mlBoostClassifier$ can only be approximated at the receiver up to a residual error (which can, however, be controlled) and thus, a guarantee in terms of the $0$-$1$-loss is not sufficient to apply Corollary~\ref{cor:boostapplication} and we instead need it to be in terms of a Lipschitz-continuous loss.

This is a relatively generic framework that can in principle work with any particular boosting technique which determines weights $\mlBoostWeight_1, \dots, \mlBoostWeight_\numUsers$ and guarantees a bound on the loss of the predictor $\mlBoostClassifier$ dependent on the errors of the base classifiers $\mlBaseClassifier{1}, \dots, \mlBaseClassifier{\numUsers}$. In the following, we describe two variations of this general approach, both based on well-known ideas from \gls{ml} (cf., e.g., \cite[Chapter 6]{mohri}).

The first one, \emph{equal majority vote}, amounts to setting $\mlBoostWeight_1 = \dots = \mlBoostWeight_\numUsers = 1$ and using local classifiers $\mlBaseClassifier{\indexUsers}$ trained only on the locally available features. This method has the advantage that the whole training procedure can be carried out in a fully decentralized way without any form of coordination or exchange of information between the agents (given that the labels for the training phase are already known everywhere; but they could, e.g., be broadcast from a central point at a cost independent of the number of agents or dimensionality of the feature space).

If we have the possibility to exchange some data between the agents, we can use the following adaptation of the AdaBoost scheme~\cite[Figure 6.1]{mohri} to the distributed setting. The algorithm runs through $\adaBoostNumIterations \leq \numUsers$ iterations, choosing a user $\adaBoostUserIndex_\adaBoostIndexIterations$ at iteration $\adaBoostIndexIterations$ to provide a base classifier $\mlBaseClassifier{\adaBoostUserIndex_\adaBoostIndexIterations}$ and assigning a corresponding weight $\mlBoostWeight_{\adaBoostUserIndex_\adaBoostIndexIterations}$. It also computes probability distributions $\adaBoostDist{1}, \dots, \adaBoostDist{\adaBoostNumIterations+1}$ on the index set of the training data $\{1, \dots, \mlNumTrainingSamples\}$, initializing $\adaBoostDist{1}$ as the uniform distribution, as well as base classifier errors $\adaBoostBaseError{1}, \dots, \adaBoostBaseError{\adaBoostNumIterations}$ and normalization constants $\adaBoostNormalization{1}, \dots, \adaBoostNormalization{\adaBoostNumIterations}$. Each iteration $\adaBoostIndexIterations$ consists of the following steps:
\begin{enumerate}
 \item \label{enum:disttraining-choice} The central node chooses a user $\adaBoostUserIndex_\adaBoostIndexIterations$ and broadcasts the choice.
 \item \label{enum:disttraining-broadcast} User $\adaBoostUserIndex_\adaBoostIndexIterations$ trains a base classifier $\mlBaseClassifier{\adaBoostUserIndex_\adaBoostIndexIterations}: \mlInputAlphabet_{\adaBoostUserIndex_\adaBoostIndexIterations} \rightarrow \{-1,1\}$ on the training sample with distribution $\adaBoostDist{\adaBoostIndexIterations}$ and broadcasts the indices of the training samples incorrectly classified by $\mlBaseClassifier{\adaBoostUserIndex_\adaBoostIndexIterations}$.
 \item From this information, every node in the system computes the following:
 \begin{itemize}
  \item $\adaBoostBaseError{\adaBoostIndexIterations} := \sum_{\mlIndexTrainingSamples=1}^\mlNumTrainingSamples \adaBoostDist{\adaBoostIndexIterations}(\mlIndexTrainingSamples) \indicator{\mlBaseClassifier{\adaBoostUserIndex_\adaBoostIndexIterations}(\mlInputAlphabetElement_{\adaBoostUserIndex_\adaBoostIndexIterations}^{(\mlIndexTrainingSamples)}) \neq \mlLabelAlphabetElement^{(\mlIndexTrainingSamples)}}$
  \item $\mlBoostWeight_{\adaBoostUserIndex_\adaBoostIndexIterations} := \frac{1}{2} \log \frac{1-\adaBoostBaseError{\adaBoostIndexIterations}}{\adaBoostBaseError{\adaBoostIndexIterations}}$
  \item $\adaBoostNormalization{\adaBoostIndexIterations} := 2\sqrt{\adaBoostBaseError{\adaBoostIndexIterations}(1-\adaBoostBaseError{\adaBoostIndexIterations}))}$
  \item $\adaBoostDist{\adaBoostIndexIterations+1}(\mlIndexTrainingSamples) := \adaBoostDist{\adaBoostIndexIterations}(\mlIndexTrainingSamples) \exp(-\mlBoostWeight_{\adaBoostUserIndex_\adaBoostIndexIterations} \mlBaseClassifier{\adaBoostUserIndex_\adaBoostIndexIterations}(\mlInputAlphabetElement_{\adaBoostUserIndex_\adaBoostIndexIterations}^{(\mlIndexTrainingSamples)}) \mlLabelAlphabetElement^{(\mlIndexTrainingSamples)})/\adaBoostNormalization{\adaBoostIndexIterations}$
 \end{itemize}
\end{enumerate}
The resulting classifier is then as defined in (\ref{eq:boostclassifier}), where we assign $\mlBoostWeight_\indexUsers := 0$ whenever $\indexUsers \neq \adaBoostUserIndex_\adaBoostIndexIterations$ for all $\adaBoostIndexIterations$. \cite[Theorem 6.1]{mohri} guarantees that the empirical $0$-$1$-loss of $\mlBoostClassifier$ is at most
\begin{align}
\label{eq:adaboostbopund}
\exp\left(
  -2
  \sum_{\adaBoostIndexIterations=1}^\adaBoostNumIterations \left( \frac{1}{2} - \adaBoostBaseError{\adaBoostIndexIterations} \right)^2
\right),
\end{align}
which unfortunately is insufficient to apply Corollary~\ref{cor:boostapplication}, because the $0$-$1$-loss is not Lipschitz-continuous. However, the proof of the theorem relies only on the inequality $\indicator{\mlBoostClassifier(\mlInputAlphabetElement^\numUsers) \mlLabelAlphabetElement \leq 0} \leq \exp(-\mlBoostClassifier(\mlInputAlphabetElement^\numUsers) \mlLabelAlphabetElement)$ for the instantaneous $0$-$1$-loss. Since the inequality
$
\log(1+\exp(-\mlBoostClassifier(\mlInputAlphabetElement^\numUsers) \mlLabelAlphabetElement))
\leq
\exp(-\mlBoostClassifier(\mlInputAlphabetElement^\numUsers) \mlLabelAlphabetElement)
$
also clearly holds, we can replace the $0$-$1$-loss in the proof with the logistic loss $\mlLoss(\mlInputAlphabetElement^\numUsers,\mlLabelAlphabetElement,\hat{\mlLabelAlphabetElement}) := \log(1+\exp(-\mlLabelAlphabetElement \hat{\mlLabelAlphabetElement}))$ (or, indeed, any other loss which satisfies this inequality). This yields the same bound (\ref{eq:adaboostbopund}) on the $1$-Lipschitz-continuous logistic loss and thus we can apply Corollary~\ref{cor:boostapplication} with $\mlLossLipschitz := 1$ to derive a guarantee on the logistic loss of the distributed approximation of our AdaBoost classifier.

We conclude with some remarks on the distributed training. The choice in step \ref{enum:disttraining-choice} could, e.g., be predetermined (in which case no communication in this step is necessary) or random, but we could also greedily select the classifier with smallest error using an instance of ScalableMax~\cite{agrawal2019scalable}\cite[Section IV]{bjelakovic2019distributed}. As for the communication cost of the distributed training, step \ref{enum:disttraining-choice} exhibits a favorable scaling which is linear in $\adaBoostNumIterations$ and logarithmic in $\numUsers$, however, step \ref{enum:disttraining-broadcast} has a cost linear in the number of training samples. There is a conceptually simpler alternative to this distributed scheme in which we communicate the full training set to the central node and perform the training in a centralized manner. The advantage in communication cost of the distributed scheme over this centralized alternative is only a constant factor. On the other hand, since only one bit per training sample and user is transmitted, this constant gain could potentially be quite large, depending on the complexity of the feature spaces. Also, in the distributed training scheme, the computational load of training the base classifiers is distributed across all nodes, which may in practice also be an advantage wherever the computational capabilities of the central node are limited. However, since the distributed training currently leverages no \gls{ota} computation and leaves that for the computation of the trained classifier itself, finding a distributed scheme which can exploit \gls{ota} computation to achieve a gain in asymptotic behavior as opposed to only a constant factor could be a worthwhile question for future research.

\subsection{Numerical Results for \gls{ota}-Computed Decision Tree Classifiers}
\label{sec:vfl-classification-numerical}
In order to illustrate how the scheme analyzed in this work can be used to compute classifiers for anomaly detection problems in large sensor networks, we have conducted numerical simulations on a synthetic binary classification problem generated by the \texttt{make\_classification} function in the \emph{datasets} package of the scikit-learn toolbox~\cite{scikit-learn} for Python. It places clusters for the two classes at the edges of a hypercube in a Euclidian space of informative features, adds redundant features that are linear combinations as well as useless features that are pure noise and applies various kinds of noise and nonlinearities. The resulting features are then shuffled randomly, partitioned and assigned to the distributed agents. For the training set, the agents also learn the correct corresponding labels. We construct two different \gls{ota}-computed classifiers as described in the preceding subsection:
\begin{enumerate}
 \item For the equal majority vote classifier, each agent trains a Decision Tree model of height at most $2$ based on the locally available features only. The \gls{ota}-computed classifier is then as put forth in equation~(\ref{eq:boostclassifier}), where $\mlBoostWeight_1 = \dots = \mlBoostWeight_\numUsers = 1$.
 \item For the AdaBoost classifier, the agents train their models cooperatively as described in the preceding section, using Decision Tree classifiers as the local base classifiers. The next agent at each iteration is picked uniformly at random from among the agents which have not yet been selected. This procedure yields not only differently trained local models compared to the equal majority vote, but also weights $\mlBoostWeight_1, \dots, \mlBoostWeight_\numUsers$ which can be used for the \gls{ota}-computed classifier as in equation~(\ref{eq:boostclassifier}).
\end{enumerate}
We assume that the agents are connected to a central receiver through a fast-fading wireless multiple-access channel, where no instantaneous \gls{csi} is available. The only kind of information we assume is available is the average power of the complex Gaussian channel gain at the transmitters and the average power of the additive noise at the receiver.
The distributed classification is simulated for noise and fading drawn from i.i.d. Gaussian distributions and for a scenario exhibiting various degrees of correlation and non-Gaussian components:
\begin{itemize}
\item For the fading, we achieve this by passing the fading coefficients through a lowpass filter, which cuts off all but a given percentage of the energy (the cutoff percentage) and then re-normalizes the remaining signal.
\item For the noise, we simulate Middleton class A noise (also called impulsive noise), which is a commonly used noise model for Power Line Communications~\cite[Section 2.6.3.1]{ferreia2010power} but has also been empirically found to be a relevant phenomenon in wireless communications~\cite{middleton1993elements}. We simulate it as described in~\cite[eq. (2.49) ff.]{ferreia2010power}: In order to create one sample of noise, a random variable $\middletonIntermediateRV$ is drawn from a Poisson distribution with intensity $\middletonImpulsive$, a parameter called the \emph{impulsive index}, and then a centered Gaussian random variable with variance
\[
2\noisePower \frac{\middletonIntermediateRV/\middletonImpulsive + \middletonRatio}{1+ \middletonRatio}
\]
is drawn, where $\noisePower$ is the overall power of the noise per complex dimension and the parameter $\middletonRatio$ is called the \emph{Gaussian-to-impulsive power ratio}. Finally, a phase shift is applied drawn uniformly from the complex unit circle. This process defines non-Gaussian, but i.i.d. noise. We have therefore modified it slightly and draw one $\middletonIntermediateRV$ for every $4$ channel realizations so that we create a correlation also in the additive noise.
\end{itemize}
We simulate the computation of each of the two distributed classifiers described above in three different ways:
}
\begin{enumerate}
 \item The \gls{dfa} scheme as described in Section~\ref{sec:thproof}.
 \item \label{item:tdma_peak} A \gls{tdma} scheme with average power normalization in which only one of the agents can transmit at a time. The information is also transmitted in analog form, since each agent conveys only one bit of information and therefore digital coded schemes would not be suitable. Since the agents transmit during a much shorter time than in the \gls{dfa} scheme, we normalize their transmission power so that the average energy consumed per channel use equals that in the \gls{dfa} scheme. The only exception to this is the case when some agents have to be allocated zero channel uses, since the number of total channel uses available is smaller than the number of agents in the system. In this case, obviously, the agents allocated zero channel uses also have zero energy consumption. This scheme has a significantly higher peak transmsission power than the \gls{dfa} scheme.
 \item A \gls{tdma} scheme with peak power normalization. It works as the one under item \ref{item:tdma_peak}, but the transmission power is normalized so that it has the same peak power as the \gls{dfa} scheme, which means that it has a significantly lower average consumption.
\end{enumerate}
We also show two baselines to make the error contribution of the compared communication schemes clearer:
\begin{enumerate}
 \item a noiseless version of the majority vote classifier, and
 \item a noiseless version of the AdaBoost-inspired classifier.
\end{enumerate}

The training set consists of $50,000$ samples and the test set of $200,000$. We have generated two different binary classification problems, one for $10$ transmitters and one for $25$ transmitters.

\paragraph{Comparison of DFA and TDMA schemes for equal majority vote and AdaBoost}
\setlength{\belowcaptionskip}{-12pt}
\begin{figure}
\begin{tikzpicture}
\begin{axis}[
  xlabel={number of complex channel uses},
  ylabel={test error},
  ymin=.04,
  ymax=.3,
  xmin=1,
  xmax=250,
  yticklabel style={
    /pgf/number format/fixed,
  },
  grid,
  legend style={at={(.5,1.1)},anchor=south},
  clip=true,
]
\addplot [
          thick,
          blue,
          solid,
          name path=equal_dfa,
         ]
         table[x=channel_uses, y=equal_dfa, col sep=comma]
         {num_results_voting_classifiers_correlated/corfading_middleton_noise_bl4.csv};
\addplot [
          thick,
          red,
          solid,
          name path=ada_dfa,
         ]
         table[x=channel_uses, y=ada_dfa, col sep=comma]
         {num_results_voting_classifiers_correlated/corfading_middleton_noise_bl4.csv};
\addplot [
          thick,
          blue,
          dashed,
          name path=equal_tdma,
         ]
         table[x=channel_uses, y=equal_tdma, col sep=comma]
         {num_results_voting_classifiers_correlated/corfading_middleton_noise_bl4.csv};
\addplot [
          thick,
          red,
          dashed,
          name path=ada_tdma,
         ]
         table[x=channel_uses, y=ada_tdma, col sep=comma]
         {num_results_voting_classifiers_correlated/corfading_middleton_noise_bl4.csv};
\addplot [thick,
          blue,
          dashdotted,
          name path=equal_tdma_peak
         ]
         table[x=channel_uses, y=equal_tdma_peak, col sep=comma]
         {num_results_voting_classifiers_correlated/corfading_middleton_noise_bl4.csv};
\addplot [thick,
          red,
          dashdotted,
          name path=ada_tdma_peak,
         ]
         table[x=channel_uses, y=ada_tdma_peak, col sep=comma]
         {num_results_voting_classifiers_correlated/corfading_middleton_noise_bl4.csv};
\addplot [very thick,red,densely dotted, domain=1:250] {0.042495};
\addplot [very thick,blue,densely dotted, domain=1:250] {0.0656};
\path [name path=gain_dfa_ada_avg] (axis cs: 1, .07) -- (axis cs: 250, .07);
\path [name intersections={of=gain_dfa_ada_avg and ada_dfa,by=intersection_chuse_ada_dfa}];
\path [name intersections={of=gain_dfa_ada_avg and ada_tdma, by=intersection_chuse_ada_tdma}];
\draw[<->, thick, red] (intersection_chuse_ada_dfa) -- (intersection_chuse_ada_tdma) coordinate[midway] (dfa_gain);
\path [name path=gain_dfa_ada_peak] (axis cs: 1, .27) -- (axis cs: 250, .27);
\path [name intersections={of=gain_dfa_ada_peak and ada_dfa,by=intersection_chuse_ada_dfa_peak}];
\path [name intersections={of=gain_dfa_ada_peak and ada_tdma_peak, by=intersection_chuse_ada_tdma_peak}];
\draw[<->, thick, red] (intersection_chuse_ada_dfa_peak) -- (intersection_chuse_ada_tdma_peak) coordinate[midway] (dfa_gain_peak);
\path [name path=gain_dfa_equal_avg] (axis cs: 1, .072) -- (axis cs: 250, .072);
\path [name intersections={of=gain_dfa_equal_avg and equal_dfa,by=intersection_chuse_equal_dfa}];
\path [name intersections={of=gain_dfa_equal_avg and equal_tdma, by=intersection_chuse_equal_tdma}];
\draw[<->, thick, blue] (intersection_chuse_equal_dfa) -- (intersection_chuse_equal_tdma) coordinate[midway] (equal_dfa_gain);
\path [name path=gain_dfa_equal_peak] (axis cs: 1, .2) -- (axis cs: 250, .2);
\path [name intersections={of=gain_dfa_equal_peak and equal_dfa,by=intersection_chuse_equal_dfa_peak}];
\path [name intersections={of=gain_dfa_equal_peak and equal_tdma_peak, by=intersection_chuse_equal_tdma_peak}];
\draw[<->, thick, blue] (intersection_chuse_equal_dfa_peak) -- (intersection_chuse_equal_tdma_peak) coordinate[midway] (equal_dfa_gain_peak);

\node[blue] (equal_dfa_gain_label) at (axis cs: 60, .15) {DFA gain};
\draw[->, thick, blue] (axis cs: 46, .14) -- (equal_dfa_gain);
\draw[->, thick, blue] (axis cs: 46, .16) -- (equal_dfa_gain_peak);

\node[red] (ada_dfa_gain_label) at (axis cs: 180, .17) {DFA gain};
\draw[->, thick, red] (axis cs: 180, .18) -- (dfa_gain_peak);
\draw[->, thick, red] (axis cs: 180, .16) -- (dfa_gain);

\legend{equal majority DFA, AdaBoost DFA, equal majority TDMA average-normalized, AdaBoost TDMA average-normalized, equal majority TDMA peak-normalized, AdaBoost TDMA peak-normalized, AdaBoost noiseless, equal majority noiseless};
\end{axis}
\end{tikzpicture}
\caption{Comparison of the classification error on the test set of DFA/TDMA and equal majority/AdaBoost schemes. $10$ transmitters, cutoff percentage $80\%$, $\middletonImpulsive = 3$, $\middletonRatio = 3$, SNR $-6$ dB.}
\label{fig:10tx-chuses}
\end{figure}
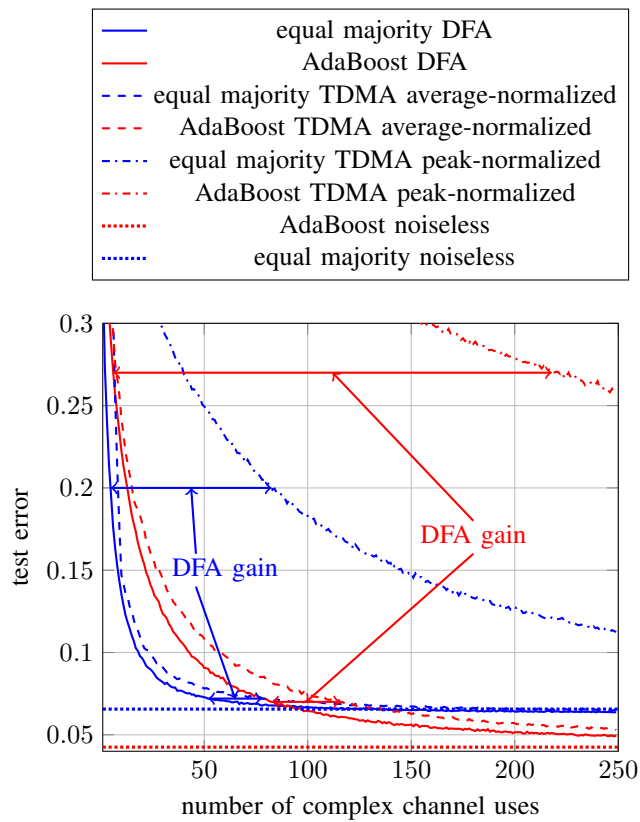

\begin{figure}
\begin{tikzpicture}
\begin{axis} [
  xlabel={SNR},
  ylabel={test error},
  ymin=.04,
  ymax=.15,
  xmin=-20,
  xmax=20,
  yticklabel style={
    /pgf/number format/fixed,
  },
  grid,
]
\addplot [
          thick,
          blue,
          solid,
          name path=equal_dfa,
         ]
         table[x=power_dB, y=equal_dfa, col sep=comma]
         {num_results_voting_classifiers_correlated/corfading_middleton_noise_bl4_snr.csv};
\addplot [
          thick,
          red,
          solid,
          name path=ada_dfa,
         ]
         table[x=power_dB, y=ada_dfa, col sep=comma]
         {num_results_voting_classifiers_correlated/corfading_middleton_noise_bl4_snr.csv};
\addplot [
          thick,
          blue,
          dashed,
          name path=equal_tdma,
         ]
         table[x=power_dB, y=equal_tdma, col sep=comma]
         {num_results_voting_classifiers_correlated/corfading_middleton_noise_bl4_snr.csv};
\addplot [
          thick,
          red,
          dashed,
          name path=ada_tdma
         ]
         table[x=power_dB, y=ada_tdma, col sep=comma]
         {num_results_voting_classifiers_correlated/corfading_middleton_noise_bl4_snr.csv};
\addplot [
          thick,
          blue,
          dashdotted,
         ]
         table[x=power_dB, y=equal_tdma_peak, col sep=comma]
         {num_results_voting_classifiers_correlated/corfading_middleton_noise_bl4_snr.csv};
\addplot [
          thick,
          red,
          dashdotted,
         ]
         table[x=power_dB, y=ada_tdma_peak, col sep=comma]
         {num_results_voting_classifiers_correlated/corfading_middleton_noise_bl4_snr.csv};
\addplot [very thick,red,densely dotted, domain=-40:40] {0.042495};
\addplot [very thick,blue,densely dotted, domain=-40:40] {0.0656};

\path [name path=gain_dfa_equal] (axis cs: 16, .05) -- (axis cs: 16, .15);
\path [name intersections={of=gain_dfa_equal and equal_dfa,by=intersection_chuse_equal_dfa}];
\path [name intersections={of=gain_dfa_equal and equal_tdma, by=intersection_chuse_equal_tdma}];
\draw[<->, thick, blue] (intersection_chuse_equal_dfa) -- (intersection_chuse_equal_tdma) coordinate[midway] (equal_dfa_gain);
\node[blue,align=center] (dfa_gain_label) at (axis cs: 10, .12) {DFA error\\floor gain};
\draw[->, thick, blue] (dfa_gain_label.south) -- (equal_dfa_gain);
\path [name path=gain_dfa_ada] (axis cs: 10, .05) -- (axis cs: 10, .15);
\path [name intersections={of=gain_dfa_ada and ada_dfa,by=intersection_chuse_ada_dfa}];
\path [name intersections={of=gain_dfa_ada and ada_tdma, by=intersection_chuse_ada_tdma}];
\draw[<->, thick, red] (intersection_chuse_ada_dfa) -- (intersection_chuse_ada_tdma) coordinate[midway] (ada_dfa_gain);
\node[red,align=center] (dfa_gain_label) at (axis cs: -10, .05) {DFA error floor gain};
\draw[->, thick, red] (dfa_gain_label.east) -- (ada_dfa_gain);

\end{axis}
\end{tikzpicture}
\caption{Comparison of the classification error on the test set of DFA/TDMA and equal majority/AdaBoost schemes. $10$ transmitters, cutoff percentage $80\%$, $\middletonImpulsive = 3$, $\middletonRatio = 3$, $50$ complex channel uses.}
\label{fig:10tx-snr}
\end{figure}
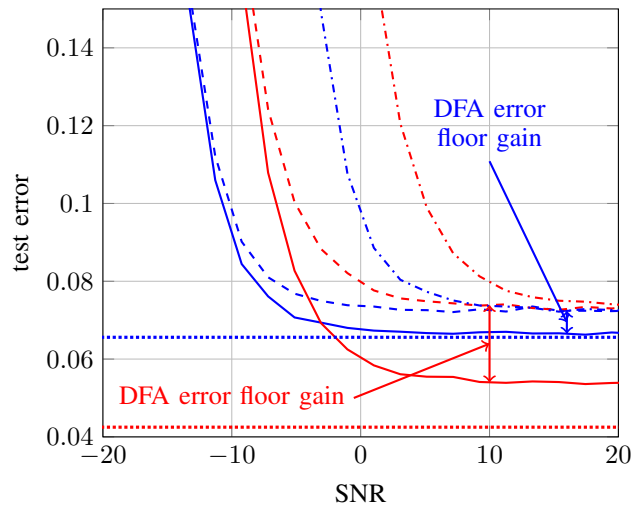

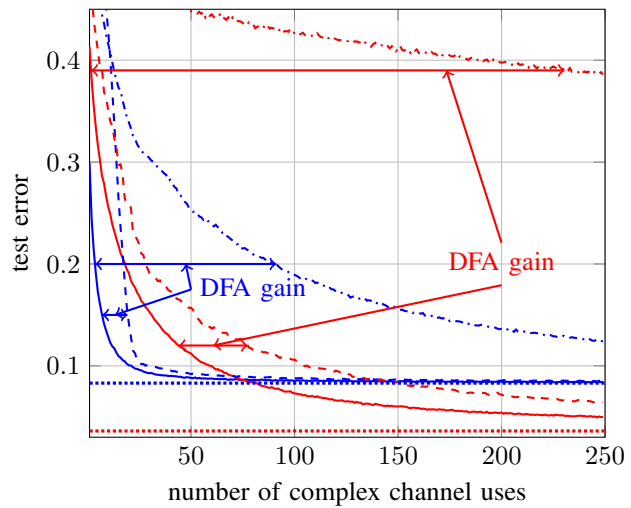
\begin{figure}
\begin{tikzpicture}
\begin{axis}[
  xlabel={number of complex channel uses},
  ylabel={test error},
  ymin=.03,
  ymax=.45,
  xmin=1,
  xmax=250,
  yticklabel style={
    /pgf/number format/fixed,
  },
  grid,
]
\addplot [
          thick,
          blue,
          solid,
          name path=equal_dfa,
         ]
         table[x=channel_uses, y=equal_dfa, col sep=comma]
         {num_results_voting_classifiers_correlated/corfading_middleton_noise_bl4_25tx.csv};
\addplot [
          thick,
          red,
          solid,
          name path=ada_dfa
         ]
         table[x=channel_uses, y=ada_dfa, col sep=comma]
         {num_results_voting_classifiers_correlated/corfading_middleton_noise_bl4_25tx.csv};
\addplot [
          thick,
          blue,
          dashed,
          name path=equal_tdma,
         ]
         table[x=channel_uses, y=equal_tdma, col sep=comma]
         {num_results_voting_classifiers_correlated/corfading_middleton_noise_bl4_25tx.csv};
\addplot [
          thick,
          red,
          dashed,
          name path=ada_tdma
         ]
         table[x=channel_uses, y=ada_tdma, col sep=comma]
         {num_results_voting_classifiers_correlated/corfading_middleton_noise_bl4_25tx.csv};
\addplot [
          thick,
          blue,
          dashdotted,
          name path=equal_tdma_peak
         ]
         table[x=channel_uses, y=equal_tdma_peak, col sep=comma]
         {num_results_voting_classifiers_correlated/corfading_middleton_noise_bl4_25tx.csv};
\addplot [
          thick,
          red,
          dashdotted,
          name path=ada_tdma_peak
         ]
         table[x=channel_uses, y=ada_tdma_peak, col sep=comma]
         {num_results_voting_classifiers_correlated/corfading_middleton_noise_bl4_25tx.csv};
\addplot [very thick,red,densely dotted, domain=1:250] {0.03612};
\addplot [very thick,blue,densely dotted, domain=1:250] {0.0831};

\path [name path=gain_dfa_ada_avg] (axis cs: 1, .12) -- (axis cs: 250, .12);
\path [name intersections={of=gain_dfa_ada_avg and ada_dfa,by=intersection_chuse_ada_dfa}];
\path [name intersections={of=gain_dfa_ada_avg and ada_tdma, by=intersection_chuse_ada_tdma}];
\draw[<->, thick, red] (intersection_chuse_ada_dfa) -- (intersection_chuse_ada_tdma) coordinate[midway] (dfa_gain);
\path [name path=gain_dfa_ada_peak] (axis cs: 1, .39) -- (axis cs: 250, .39);
\path [name intersections={of=gain_dfa_ada_peak and ada_dfa,by=intersection_chuse_ada_dfa_peak}];
\path [name intersections={of=gain_dfa_ada_peak and ada_tdma_peak, by=intersection_chuse_ada_tdma_peak}];
\draw[<->, thick, red] (intersection_chuse_ada_dfa_peak) -- (intersection_chuse_ada_tdma_peak) coordinate[near end] (dfa_gain_peak);
\path [name path=gain_dfa_equal_avg] (axis cs: 1, .15) -- (axis cs: 250, .15);
\path [name intersections={of=gain_dfa_equal_avg and equal_dfa,by=intersection_chuse_equal_dfa}];
\path [name intersections={of=gain_dfa_equal_avg and equal_tdma, by=intersection_chuse_equal_tdma}];
\draw[<->, thick, blue] (intersection_chuse_equal_dfa) -- (intersection_chuse_equal_tdma) coordinate[midway] (equal_dfa_gain);
\path [name path=gain_dfa_equal_peak] (axis cs: 1, .2) -- (axis cs: 250, .2);
\path [name intersections={of=gain_dfa_equal_peak and equal_dfa,by=intersection_chuse_equal_dfa_peak}];
\path [name intersections={of=gain_dfa_equal_peak and equal_tdma_peak, by=intersection_chuse_equal_tdma_peak}];
\draw[<->, thick, blue] (intersection_chuse_equal_dfa_peak) -- (intersection_chuse_equal_tdma_peak) coordinate[midway] (equal_dfa_gain_peak);

\node[blue,align=left] (equal_dfa_gain_label) at (axis cs: 80, .175) {DFA gain};
\draw[->, thick, blue] (equal_dfa_gain_label.west) -- (equal_dfa_gain);
\draw[->, thick, blue] (equal_dfa_gain_label.west) -- (equal_dfa_gain_peak);

\node[red,align=center] (ada_dfa_gain_label) at (axis cs: 200, .2) {DFA gain};
\draw[->, thick, red] (ada_dfa_gain_label.north) -- (dfa_gain_peak);
\draw[->, thick, red] (ada_dfa_gain_label.south) -- (dfa_gain);
\end{axis}
\end{tikzpicture}
\caption{Simulation results for $25$ transmitters at a fixed SNR of $-6$ dB, correlation parameters are the same as in Fig.~\ref{fig:10tx-chuses}.}
\label{fig:25tx-chuses}
\end{figure}

We have simulated the DFA scheme as well as the TDMA baseline comparisons for a scenario with moderate correlation and non-Gaussianity. The cutoff percentage for the lowpass filter applied to the fading was chosen at $80 \%$, and the parameters for the Middleton Class A noise were $\middletonImpulsive = 3$ and $\middletonRatio = 3$. In Fig.~\ref{fig:10tx-chuses}, we plot the classification error on the test set as a function of the number of complex channel uses for a fixed SNR of $-6$ dB and $10$ tranmitters, and in Fig.~\ref{fig:10tx-snr}, we plot the error as function of SNR for a fixed number of $50$ complex channel uses. We can see that, as the effect of the multiplicative fading dominates that of the additive noise, the schemes reach an error floor that cannot be lowered with an increase of the transmission power. When the number of complex channel uses is increased, on the other hand, the error curves approach the noiseless classification error even if the SNR is kept fixed.

For instance, to obtain a classification error of $0.07$ or better, both for the equal majority vote and AdaBoost, the average-power normalized \gls{tdma} scheme needs over $30$ channel uses more than the \gls{dfa} scheme. Since we compare with a \gls{tdma} scheme that uses the same average energy per channel use as the \gls{dfa} scheme, this means that the \gls{tdma} scheme not only consumes more wireless spectrum and/or time, but also significantly more energy. For the case of the same peak power consumption (which means that \gls{tdma} consumes less average power since transmitters are silent most of the time), the difference is huge and can be several hundred channel uses depending on the error level.

The advantage of the \gls{dfa} scheme over the \gls{tdma} alternatives is quite pronounced even at a relatively low number of only $10$ transmitters. In Fig.~\ref{fig:25tx-chuses}, we show the same plot as in Fig.~\ref{fig:10tx-chuses}, but for a different \gls{ml} problem with $25$ transmitters. It can be seen in the plot that as the number of transmitters increases, the difference in performance between the \gls{dfa} and \gls{tdma} schemes becomes even stronger. This is due to the different scaling behaviors of the schemes.

We have run these simulations for many instantiations of the randomly generated classification schemes (not depicted for lack of space) and note that while in some cases the equal majority vote scheme performs similarly as the AdaBoost scheme in the noiseless case, in many cases the error behavior of the AdaBoost scheme is much better and it is more robust, e.g. in the case that a large number of agents observes only useless features while only few agents observe the informative and repetitive features that can be used to solve the classification problem. That being said, in the case in which the equal majority vote performs similarly to AdaBoost in the noiseless case, its error behavior in the communication schemes is better since it better utilizes the available peak transmission power.

\paragraph{Synchronization errors}
\label{par:sync-errors}
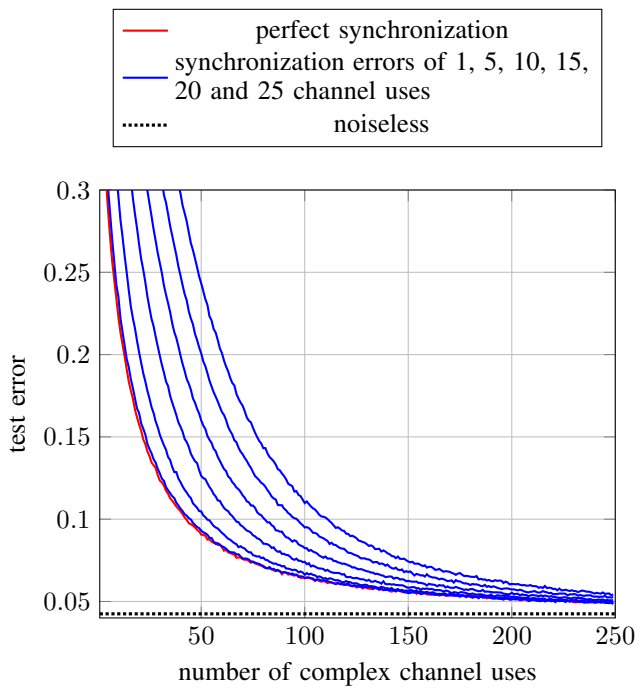
\begin{figure}
\begin{tikzpicture}
\begin{axis}[
  xlabel={number of complex channel uses},
  ylabel={test error},
  ymin=.04,
  ymax=.3,
  xmin=1,
  xmax=250,
  yticklabel style={
    /pgf/number format/fixed,
  },
  grid,
  legend style={at={(.5,1.1)},anchor=south,cells={align=left}},
  clip=true,
]
\addplot [
          thick,
          red,
          solid,
          name path=ada_dfa,
         ]
         table[x=channel_uses, y=ada_dfa, col sep=comma]
         {num_results_voting_classifiers_correlated/corfading_middleton_noise_bl4.csv};
\addplot [
          thick,
          blue,
          solid,
          name path=ada_dfa,
         ]
         table[x=channel_uses, y=ada_dfa, col sep=comma]
         {num_results_voting_classifiers_correlated/corfading_middleton_noise_bl4_unsync_1.csv};
\addplot [very thick,black,densely dotted, domain=1:600] {0.042495};
\addplot [
          thick,
          blue,
          solid,
          name path=ada_dfa,
         ]
         table[x=channel_uses, y=ada_dfa, col sep=comma]
         {num_results_voting_classifiers_correlated/corfading_middleton_noise_bl4_unsync_5.csv};
\addplot [
          thick,
          blue,
          solid,
          name path=ada_dfa,
         ]
         table[x=channel_uses, y=ada_dfa, col sep=comma]
         {num_results_voting_classifiers_correlated/corfading_middleton_noise_bl4_unsync_10.csv};
\addplot [
          thick,
          blue,
          solid,
          name path=ada_dfa,
         ]
         table[x=channel_uses, y=ada_dfa, col sep=comma]
         {num_results_voting_classifiers_correlated/corfading_middleton_noise_bl4_unsync_15.csv};
\addplot [
          thick,
          blue,
          solid,
          name path=ada_dfa,
         ]
         table[x=channel_uses, y=ada_dfa, col sep=comma]
         {num_results_voting_classifiers_correlated/corfading_middleton_noise_bl4_unsync_20.csv};
\addplot [
          thick,
          blue,
          solid,
          name path=ada_dfa,
         ]
         table[x=channel_uses, y=ada_dfa, col sep=comma]
         {num_results_voting_classifiers_correlated/corfading_middleton_noise_bl4_unsync_25.csv};
\legend{perfect synchronization, {synchronization errors of 1, 5, 10, 15, \\20 and 25 channel uses}, noiseless};
\end{axis}
\end{tikzpicture}
\caption{Impact of synchronization errors on the DFA AdaBoost scheme in the scenario of Fig.~\ref{fig:10tx-chuses}.}
\label{fig:sync-errors}
\end{figure}
Since the pre-processing described in Section~\ref{sec:pre-proc} creates a sequence of i.i.d. random variables (conditioned under $s_1, \dots, s_K$), we can expect that the scheme is quite robust to synchronization errors between the transmitters. This is important since perfect synchronization of the transmitted signals at the receiver would be a very hard task to achieve in practice. In order to substantiate this argument, we have run simulations with relatively large synchronization errors of several symbol durations: For the starting time of the transmission for each transmitter, we have added a uniformly random number of channel uses in a range of up to 1, up to 5, up to 10, up to 15, up to 20 and up to 25 channel uses. In Fig.~\ref{fig:sync-errors}, we show the impact of the synchronization errors for the AdaBoost \gls{dfa} scheme and the same choice of parameters as for Fig.~\ref{fig:10tx-chuses}. The solid red curves (representing the case of perfect synchronization) are therefore the same in both figures while the blue curves in Fig.~\ref{fig:sync-errors} depict the performance for various values of the maximum synchronization error. The performance degrades gracefully even for extremely large synchronization errors and the number of additional channel uses that needs to be expended to maintain the same classification error is about twice the value of the maximum synchronization error. We remark that we expect the number of additional channel uses that needs to be expended to scale with the synchronization error and not with the number of transmitters. Moreover, for synchronization errors of the same order of magnitude as the symbol duration, the drop in performance is barely noticable.

\paragraph{Comparison of different correlation scenarios}

\begin{figure}
\begin{tikzpicture}
\begin{axis}[
  xlabel={number of complex channel uses},
  ylabel={test error},
  ymin=.04,
  ymax=.3,
  xmin=1,
  xmax=250,
  yticklabel style={
    /pgf/number format/fixed,
  },
  grid,
  legend style={at={(.5,1.1)},anchor=south,cells={align=left}},
  clip=true,
]
\addplot [
          thick,
          green,
          solid,
          name path=ada_dfa,
         ]
         table[x=channel_uses, y=ada_dfa, col sep=comma]
         {num_results_voting_classifiers_correlated/iid.csv};
\addplot [
          thick,
          red,
          solid,
          name path=ada_dfa,
         ]
         table[x=channel_uses, y=ada_dfa, col sep=comma]
         {num_results_voting_classifiers_correlated/corfading_middleton_noise_bl4.csv};
\addplot [
          thick,
          blue,
          solid,
          name path=ada_dfa,
         ]
         table[x=channel_uses, y=ada_dfa, col sep=comma]
         {num_results_voting_classifiers_correlated/corfading0.8_middleton_noise_impulsive0.5_ratio1_bl4.csv};
\addplot [
          thick,
          brown,
          solid,
          name path=ada_dfa,
         ]
         table[x=channel_uses, y=ada_dfa, col sep=comma]
         {num_results_voting_classifiers_correlated/corfading0.5_middleton_noise_impulsive0.25_ratio1_bl4.csv};
\addplot [
          thick,
          black,
          solid,
          name path=ada_dfa,
         ]
         table[x=channel_uses, y=ada_dfa, col sep=comma]
         {num_results_voting_classifiers_correlated/corfading0.3_middleton_noise_impulsive1_ratio0.5_bl4.csv};
\addplot [very thick,black,densely dotted, domain=1:600] {0.042495};

\legend{{i.i.d. Gaussian fading and noise},{cutoff percentage $80\%$, $\middletonImpulsive = 3$, $\middletonRatio = 3$\\(same as Fig.~\ref{fig:10tx-chuses})},{cutoff percentage $80\%$, $\middletonImpulsive = 0.5$, $\middletonRatio = 1$},{cutoff percentage $50\%$, $\middletonImpulsive = 0.25$, $\middletonRatio = 1$},{cutoff percentage $30\%$, $\middletonImpulsive = 1$, $\middletonRatio = 0.5$}, noiseless};
\end{axis}
\end{tikzpicture}
\caption{Comparison of the performance of AdaBoost DFA for various choices of the non-Gaussianity and correlation parameters.}
\label{fig:correlation-comparison}
\end{figure}
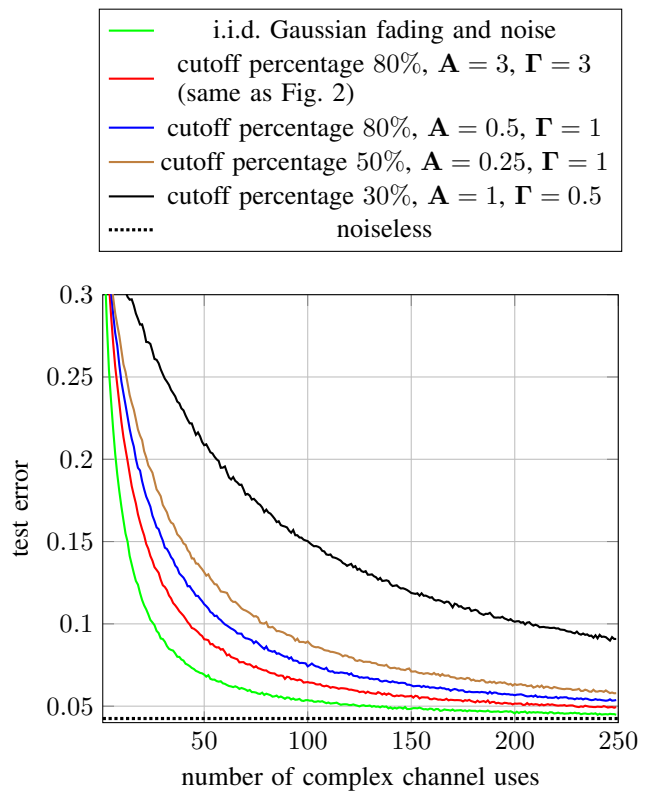
In order to get an idea of how strongly the correlation and non-Gaussianity impact on the performance of the scheme, we have compared the AdaBoost DFA scheme (the solid red curve from Fig.~\ref{fig:10tx-chuses}) for various choices of the correlation and non-Gaussianity parameters. Qualitatively, the higher the values of $\middletonImpulsive$ and $\middletonRatio$, the more closely the additive noise is to a Gaussian distribution and the lower the values, the stronger pronounced is the non-Gaussianity of the noise. For the fading, a cutoff percentage of $100 \%$ corresponds to i.i.d. Gaussian fading, while lower cutoff percentages mean that the fading changes more slowly over time. The results of our comparison in Fig.~\ref{fig:correlation-comparison} show that the scheme performs best for the Gaussian i.i.d. case but, as is expected from the theoretical analysis, the exponential decay of the error is retained even for strongly pronounced correlation and non-Gaussianity.

\section{Proof of Theorem~\ref{th:correlated}}
\label{sec:thproof}
\subsection{Pre-Processing}\label{sec:pre-proc}
In the pre-processing step we encode the function values $f_k(s_k)$, $k=1,\ldots ,K$ as transmit power:  
\begin{equation*}
X_k(m):= \sqrt{\transmitPower_k} U_k (m),
1\leq m\leq M
\end{equation*}
with $\transmitPower_k = g_k(f_k (s_k))$, where $g_k: [ \phi_{\min,k}, \phi_{\max,k}] \to [ 0, P  ] $
such that
\begin{equation}\label{eq:pre-processing}
  g_k(t):= \frac{P}{\Delta(f)} (t-\phi_{\min,k}),
  \end{equation}
 where $\Delta(f)$ is given in (\ref{eq:max-inner-spread}) and   $\phi_{\min,k}$ is defined in (\ref{eq:phi-def-spread}).\\
 $U_k (m)$, $k=1, \ldots, K$, $m=1, \ldots ,M$ are i.i.d. with the uniform distribution on $\{-1,+1\}$. We assume the random variables $U_k (m)$, $k=1,\ldots,K$, $m=1\ldots,M$, are independent of
 $H_k(m)$, $k=1,\ldots, K$, $m=1,\ldots,M$, and $N(m)$, $m=1,\ldots,M$.

We write the vector of transmitted signals at channel use $\indexChannelUses$ as
\begin{align*}
\normalizedMessages(\indexChannelUses) := \left(X_1(m),\dots,X_K(m)\right)
\end{align*}
and combine them in a matrix $\messagesMatrix :=$
\begin{align*}
\begin{pmatrix}
\normalizedMessages(1) & 0                      & 0                      & 0                      & 0      & \dots                                  & 0                                      \\
0                      & \normalizedMessages(1) & 0                      & 0                      & 0      & \dots                                  & 0                                      \\
0                      & 0                      & \normalizedMessages(2) & 0                      & 0      & \dots                                  & 0                                      \\
0                      & 0                      & 0                      & \normalizedMessages(2) & 0      & \dots                                  & 0                                      \\
                       &                        &                        &                        & \ddots &                                        &                                        \\
0                      & 0                      & 0                      & \dots                  & 0      & \normalizedMessages(\numChannelUses) & 0                                      \\
0                      & 0                      & 0                      & \dots                  & 0      & 0                                      & \normalizedMessages(\numChannelUses)
\end{pmatrix}.
\end{align*}
%
%
\subsection{Post-Processing}\label{sec:post-proc}
The vector $\rxVector$ of received signals across the $\numChannelUses$ channel uses can be written as
$
\rxVector = \messagesMatrix \cdot \channelVector + \noiseVector,
$
where $\channelVector$ and $\noiseVector$ are given in (\ref{eq:channelvector}) and (\ref{eq:noisevector}).
The post-processing is based on receive energy which has the form
\begin{align*}
\vectorNorm{\rxVector}^2
&=
\transpose{\rxVector} \rxVector
=
\transpose{(\messagesMatrix \fadingTransformMatrix \randomBaseVector + \noiseTransformMatrix \randomBaseVector)}
(\messagesMatrix \fadingTransformMatrix \randomBaseVector + \noiseTransformMatrix \randomBaseVector)
=
\transpose{\randomBaseVector}
\hwEffectiveMatrix
\randomBaseVector,
\end{align*}
where we use
\begin{align}
\hwEffectiveMatrix
&:=
\transpose{( \messagesMatrix \fadingTransformMatrix + \noiseTransformMatrix)}
( \messagesMatrix \fadingTransformMatrix + \noiseTransformMatrix)
\nonumber
\\
&=
\transpose{\fadingTransformMatrix} \transpose{\messagesMatrix}
\messagesMatrix \fadingTransformMatrix
+
\transpose{\fadingTransformMatrix} \transpose{\messagesMatrix}
\noiseTransformMatrix
+
\transpose{\noiseTransformMatrix}
\messagesMatrix \fadingTransformMatrix
+
\transpose{\noiseTransformMatrix}
\noiseTransformMatrix.
\label{eq:hwEffectiveExpression}
\end{align}

Equivalently, we can phrase this as
\begin{equation}
\label{eq:energy}  
\vectorNorm{\rxVector}^2
=
\sum_{k=1}^K  \transmitPower_k \| H_k \|_2^2
+ \bar{N}_{s^{K}},
\end{equation}
where $H_k=(H_k(1), \ldots , H_k(M))$ is a vector consisting of complex fading coefficients, and
 $\bar{N}_{s^{K}}= \sum_{m=1}^M \bar{N}_{s^K}(m)$. The random variables $\bar{N}_{s^K}(m) $, $ m=1, \ldots , M$, are given by
\begin{align}
\bar{N}_{s^K} (m) := 
  &\hphantom{+} \sum_{\substack{k ,l=1, \\ k\neq l}}^K \sqrt{\transmitPower_k \transmitPower_l} H_k(m)\overline{H_l (m)} \nonumber \\
  &\hphantom{+\sum_{\substack{k ,l=1, \\ k\neq l}}^K} \times  U_{k}(m)U_{l}(m)  \nonumber \displaybreak[0] \\ 
  & + 2 \Real \left( \overline{N(m)} \sum_{k=1}^K \sqrt{\transmitPower_k} H_k(m)U_{k}(m) \right ) \nonumber \\
  & + | N(m) |^2.
\label{eq:N-s-k}
\end{align}
The receiver computes its estimate $\bar{f}$ of $f(s_1, \dots, s_K)$ as
\begin{align*}
\bar{f}
:=
F(\bar{g}(
  \vectorNorm{\rxVector}^2
  -
  \Expectation \vectorNorm{\noiseVector}^2
)),
\end{align*}
where
\[
\bar{g}(t) := \frac{\Delta(f)}{2 \cdot M \cdot  P}t + \sum_{k=1}^K\phi_{\min,k}.
\]

\subsection{The Error Event}\label{sec:error-event}
Clearly, $\Expectation \bar{N}_{s^K} (m) = \Expectation | N(m) |^2$ (since all the other summands in (\ref{eq:N-s-k}) are centered). We can therefore conclude
\begin{align*}
\Expectation\left(
\bar{g}\left(
  \vectorNorm{\rxVector}^2 - \Expectation \vectorNorm{\noiseVector}^2
\right)
\right)
&=
\bar{g}\left(
  \Expectation \vectorNorm{\rxVector}^2 - \Expectation \vectorNorm{\noiseVector}^2
\right)
\\
&=
\sum_{\indexUsers=1}^\numUsers f_\indexUsers(s_\indexUsers).
\end{align*}
We use this to argue
\begin{align}
&\hphantom{{}={}}
\absolute{\bar{f}-f(s_1, \dots, s_\numUsers)}
\nonumber
\\
&=
\absolute{F\left(\bar{g}\left(\vectorNorm{\rxVector}^2 - \Expectation \vectorNorm{\noiseVector}^2\right)\right) - F\left(\sum_{\indexUsers=1}^\numUsers f_\indexUsers(s_\indexUsers)\right)}
\displaybreak[0]
\nonumber
\\
&\leq
\Phi\left(\absolute{\bar{g}\left(\vectorNorm{\rxVector}^2 - \Expectation \vectorNorm{\noiseVector}^2\right) - \sum_{\indexUsers=1}^\numUsers f_\indexUsers(s_\indexUsers)}\right)
\displaybreak[0]
\nonumber
\\
&=
\Phi\left(
  \absolute{\bar{g}\left(\vectorNorm{\rxVector}^2 - \Expectation \vectorNorm{\noiseVector}^2 \right)
  -
  \bar{g}\left(
    \Expectation \vectorNorm{\rxVector}^2 - \Expectation \vectorNorm{\noiseVector}^2
  \right)}
\right)
\displaybreak[0]
\nonumber
\\
&=
\Phi\left(
  \frac{\Delta(f)}{2\numChannelUses\powerconstraint}
  \absolute{\vectorNorm{\rxVector}^2 - \Expectation \vectorNorm{\rxVector}^2}
\right)
\label{eq:correlated-fY-ingredient}
\end{align}
and therefore
\begin{multline}
\label{eq:correlated-fY}
\Probability\left(
  \absolute{\bar{f}-f(s_1, \dots, s_\numUsers)}
  \geq
  \tail
\right)
\\
\leq
\Probability\left(
  \absolute{\vectorNorm{\rxVector}^2 - \Expectation \vectorNorm{\rxVector}^2}
  \geq
  \frac{2\numChannelUses\powerconstraint}{\Delta(f)}
  \Phi^{-1}(\tail)
\right).
\end{multline}

\subsection{Performance Bounds}\label{sec:performance-bounds}
Our objective is now to establish the concentration of $\vectorNorm{\rxVector}^2$ around its expectation and thus obtain an upper bound for the right hand side of (\ref{eq:correlated-fY}). To this end, we first need to establish a series of lemmas that we will use as tools.

We will split the deviation from the mean into a diagonal and an off-diagonal part. The first lemma will later help us bound the diagonal part of the error.

\begin{lemma}
\label{lemma:splitting}
Let $\generalRandomVector_1, \dots, \generalRandomVector_\rvDimension$ be independent and centered with sub-Gaussian norm at most $1$. Let $\generalTransformMatrix_1, \dots, \generalTransformMatrix_\rvDimension$ be random variables independent of $\generalRandomVector_1, \dots, \generalRandomVector_\rvDimension$ but not necessarily of each other, and assume that for all $\rvDimIndex$, $\absolute{\generalTransformMatrix_\rvDimIndex} \leq \generalUpperBound$ and $\sum_{\rvDimIndex=1}^\rvDimension \generalTransformMatrix_\rvDimIndex^2 \leq \generalSumUpperBound$ almost surely. Then we have for any $\auxVariable \in (0,1)$ and any $\mgfVariable \in (-c/(2\generalUpperBound), c/(2\generalUpperBound))$,
\begin{multline*}
\Expectation \exp\left(
  \mgfVariable
  \sum_{\rvDimIndex=1}^\rvDimension \left(
    \generalTransformMatrix_\rvDimIndex \generalRandomVector_\rvDimIndex^2
    -
    \Expectation(\generalTransformMatrix_\rvDimIndex \generalRandomVector_\rvDimIndex^2)
  \right)
\right)
\\
\leq
\exp\left(
  \frac{\mgfVariable^2}{2} \cdot \frac{8\generalSumUpperBound}{1-\auxVariable}
\right)
\Expectation\exp\left(\mgfVariable \sum_{\rvDimIndex=1}^\rvDimension \big(\generalTransformMatrix_\rvDimIndex-\Expectation(\generalTransformMatrix_\rvDimIndex)\big)\right).
\end{multline*}
\end{lemma}
\begin{proof}
The lemma follows by a straightforward calculation
\begin{align}
&\hphantom{{}={}}
\Expectation \exp\left(
  \mgfVariable
  \sum_{\rvDimIndex=1}^\rvDimension \left(
    \generalTransformMatrix_\rvDimIndex\generalRandomVector_\rvDimIndex^2
    -
    \Expectation(\generalTransformMatrix_\rvDimIndex \generalRandomVector_\rvDimIndex^2)
  \right)
\right)
\nonumber
\displaybreak[0] \\
&=
\begin{aligned}[t]
\Expectation \Bigg(
  &\exp\left(
    \mgfVariable
    \sum_{\rvDimIndex=1}^\rvDimension \left(
      \generalTransformMatrix_\rvDimIndex
        (\generalRandomVector_\rvDimIndex^2
        -
        \Expectation(\generalRandomVector_\rvDimIndex^2)
        )
    \right)
  \right)
  \\
  &\cdot\exp\left(
    \mgfVariable
    \sum_{\rvDimIndex=1}^\rvDimension \left(
      \Expectation (\generalRandomVector_\rvDimIndex^2)
      (\generalTransformMatrix_\rvDimIndex
       -
       \Expectation(\generalTransformMatrix_\rvDimIndex
      )
    \right)
  \right)
\Bigg)
\end{aligned}
\displaybreak[0]
\\
&=
\label{eq:splitting-independence}
\begin{aligned}[t]
\Expectation_\generalTransformMatrix \Bigg(
  &\exp\left(
    \mgfVariable
    \sum_{\rvDimIndex=1}^\rvDimension \left(
      \Expectation (\generalRandomVector_\rvDimIndex^2)
      (\generalTransformMatrix_\rvDimIndex
       -
       \Expectation(\generalTransformMatrix_\rvDimIndex
      )
    \right)
  \right)
  \\
  &\cdot\prod_{\rvDimIndex=1}^\rvDimension \Expectation_\generalRandomVector \exp\left(
    (\mgfVariable \generalTransformMatrix_\rvDimIndex)
    \left(\generalRandomVector_\rvDimIndex^2
    -
    \Expectation(\generalRandomVector_\rvDimIndex^2)
    \right)
  \right)
\Bigg)
\end{aligned}
\displaybreak[0]
\\
&\leq
\label{eq:splitting-mgflemma}
\begin{aligned}[t]
\Expectation_\generalTransformMatrix \Bigg(
  &\exp\left(
    \mgfVariable
    \sum_{\rvDimIndex=1}^\rvDimension \left(
      \Expectation (\generalRandomVector_\rvDimIndex^2)
      (\generalTransformMatrix_\rvDimIndex
       -
       \Expectation(\generalTransformMatrix_\rvDimIndex
      )
    \right)
  \right)
  \\
  &\cdot\prod_{\rvDimIndex=1}^\rvDimension \exp\left(
    \frac{\mgfVariable^2}{2}
    \cdot
    \frac{8\generalTransformMatrix_\rvDimIndex^2}{1-\auxVariable}
  \right)
\Bigg)
\end{aligned}
\displaybreak[0]
\\
&\leq
\label{eq:splitting-asconditions}
\exp\left(
    \frac{\mgfVariable^2}{2}
    \cdot
    \frac{8\generalSumUpperBound}{1-\auxVariable}
  \right)
\Expectation \left(
  \exp\left(
    \mgfVariable
    \sum_{\rvDimIndex=1}^\rvDimension \left(
      \generalTransformMatrix_\rvDimIndex
      -
      \Expectation\generalTransformMatrix_\rvDimIndex
    \right)
  \right)
\right),
\end{align}
where (\ref{eq:splitting-independence}) follows by the independence assumptions, (\ref{eq:splitting-mgflemma}) is an application of Lemma~\ref{lemma:app-2} and (\ref{eq:splitting-asconditions}) holds because $\sum_{\rvDimIndex=1}^\rvDimension \generalTransformMatrix_\rvDimIndex^2 \leq \generalSumUpperBound$ almost surely.
\end{proof}

The next lemma is a slight variation of the Hanson-Wright inequality as phrased in~\cite[Theorem 6.2.1]{vershynin} and will help us bound the off-diagonal part of the error.
\begin{lemma}
\label{lemma:hw-offdiagonal}
Let $\generalRandomVector$ be an $\reals^\rvDimension$-valued random variable with independent, centered entries and assume that for all $\rvDimIndex \in \{1,\dots,\rvDimension\}$, the $\rvDimIndex$-th entry of $\generalRandomVector$ satisfies $\subgaussnorm{\generalRandomVector_\rvDimIndex} \leq \subgaussBound$. Let $\generalTransformMatrix \in \reals^{\rvDimension \times \rvDimension}$ with zeros on the diagonal and $\tail > 0$. Suppose further that $\operatornorm{\generalTransformMatrix} \leq \operatornormBound$ and $\frobeniusnorm{\generalTransformMatrix} \leq \frobeniusnormBound$. Then $\Expectation\left(\transpose{\generalRandomVector} \generalTransformMatrix \generalRandomVector\right) = 0$ and
\begin{align}
\label{eq:hw-offdiagonal-statement}
\Probability\left(
  \absolute{
    \transpose{\generalRandomVector} \generalTransformMatrix \generalRandomVector
  }
  \geq
  \tail
\right)
\leq
2\exp\left(
-\frac{\tail^2}
      {16 \subgaussBound^2 \tail \operatornormBound + 256 \subgaussBound^4 \frobeniusnormBound^2}
\right).
\end{align}
\end{lemma}
This lemma differs from \cite[Theorem 6.2.1]{vershynin} mainly in that we require the diagonal entries of $\generalTransformMatrix$ to be $0$ and that all the constants are explicit. The proof follows \cite{vershynin} closely and is given in Section~\ref{appendix-hw} of the Supplement. We remark that it is not hard to follow the proof in~\cite{vershynin} further and expand the result to matrices with non-zero diagonal elements, however, this is not relevant for the present work.

Mainly because the matrix $\hwEffectiveMatrix$ contains randomness, we need a slight modification of this lemma as well as two more lemmas exploring some specific properties of $\hwEffectiveMatrix$.
\begin{cor}
Assume a setting as in Lemma~\ref{lemma:hw-offdiagonal}, but let $\generalTransformMatrix$ be an $\reals^{\rvDimension \times \rvDimension}$-valued random variable independent of $\generalRandomVector$ such that almost surely, the diagonal entries of $\generalTransformMatrix$ are $0$, $\operatornorm{\generalTransformMatrix} \leq \operatornormBound$ and $\frobeniusnorm{\generalTransformMatrix} \leq \frobeniusnormBound$. Then $\Expectation\left(\transpose{\generalRandomVector} \generalTransformMatrix \generalRandomVector\right) = 0$ and (\ref{eq:hw-offdiagonal-statement}), considering joint expectation, respectively probability of $\generalRandomVector$ and $\generalTransformMatrix$, still hold.
\end{cor}
\begin{proof}
$\Expectation\left(\transpose{\generalRandomVector} \generalTransformMatrix \generalRandomVector\right) = 0$ as well as (\ref{eq:hw-offdiagonal-statement}) hold conditional on any realization of $\generalTransformMatrix$ (except possibly in a null set) and therefore, the Corollary follows by the laws of total expectation and total probability.
\end{proof}
\begin{lemma}
\label{lemma:hwNormBound}
We have almost surely
\begin{align*}
\frobeniusnorm{\hwEffectiveMatrix}
&\leq
\begin{multlined}[t]
\left(
  \sqrt{\Delta(f || \powerconstraint)} \operatornorm{\fadingTransformMatrix} + \operatornorm{\noiseTransformMatrix}
\right)
\\
\cdot
\left(
  \sqrt{\Delta(f || \powerconstraint)} \frobeniusnorm{\fadingTransformMatrix} + \frobeniusnorm{\noiseTransformMatrix}
\right)
\end{multlined}
\displaybreak[0] \\
\operatornorm{\hwEffectiveMatrix}
&\leq
\left(
  \sqrt{\Delta(f || \powerconstraint)} \operatornorm{\fadingTransformMatrix} + \operatornorm{\noiseTransformMatrix}
\right)^2.
\end{align*}
\end{lemma}
\begin{proof}
In order to bound the norm of $\hwEffectiveMatrix$, we first note that
\begin{align}
\label{eq:correlated-QQT}
\messagesMatrix
\transpose{\messagesMatrix}
=
\sum_{\indexUsers=1}^\numUsers a_\indexUsers \identityMatrix_{2\numChannelUses}.
\end{align}
Therefore, we can conclude that all singular values of $\messagesMatrix$ are bounded by $\sqrt{\Delta(f || \powerconstraint)}$ and thus $\operatornorm{\messagesMatrix} \leq \sqrt{\Delta(f || \powerconstraint)}$.

Noting that $\frobeniusnorm{\generalmatrixOne\generalMatrixTwo} \leq \operatornorm{\generalmatrixOne} \frobeniusnorm{\generalMatrixTwo}$ for all matrices $\generalmatrixOne, \generalMatrixTwo$ of compatible dimensions and further noting the submultiplicativity of the operator norm and the triangle inequality of both norms, we get
\begin{align*}
\frobeniusnorm{\hwEffectiveMatrix}
&\leq
\operatornorm{\messagesMatrix \fadingTransformMatrix + \noiseTransformMatrix}\frobeniusnorm{\messagesMatrix \fadingTransformMatrix + \noiseTransformMatrix}
\\
&\leq
\left(
\operatornorm{\messagesMatrix} \operatornorm{\fadingTransformMatrix} + \operatornorm{\noiseTransformMatrix}
\right)
\left(
  \operatornorm{\messagesMatrix} \frobeniusnorm{\fadingTransformMatrix} + \frobeniusnorm{\noiseTransformMatrix}
\right)
\\
&\leq
\begin{multlined}[t]
\left(
  \sqrt{\Delta(f || \powerconstraint)} \operatornorm{\fadingTransformMatrix} + \operatornorm{\noiseTransformMatrix}
\right)
\\
\cdot
\left(
  \sqrt{\Delta(f || \powerconstraint)} \frobeniusnorm{\fadingTransformMatrix} + \frobeniusnorm{\noiseTransformMatrix}
\right)
\end{multlined}
\displaybreak[0]
\\
\operatornorm{\hwEffectiveMatrix}
&=
\operatornorm{\messagesMatrix \fadingTransformMatrix + \noiseTransformMatrix}^2
\leq
\left(
  \operatornorm{\messagesMatrix} \operatornorm{\fadingTransformMatrix} + \operatornorm{\noiseTransformMatrix}
\right)^2
\\
&\leq
\left(
  \sqrt{\Delta(f || \powerconstraint)} \operatornorm{\fadingTransformMatrix} + \operatornorm{\noiseTransformMatrix}
\right)^2
\qedhere
\end{align*}
\end{proof}
\begin{lemma}
\label{lemma:hwMGF}
We have
\begin{multline}
\label{eq:hwMGF}
\subgaussnorm{\trace \hwEffectiveMatrix}
\leq
4 \numChannelUses \Delta(f||\powerconstraint)
\operatornorm{(\fadingTransformMatrix+\uncorrelatedApprox)\transpose{(\fadingTransformMatrix-\uncorrelatedApprox)}}
\\
+ 2\sqrt{2 \powerconstraint} \frobeniusnorm{\fadingTransformMatrix\transpose{\noiseTransformMatrix}}.
\end{multline}
\end{lemma}
\begin{proof}
With an addition of zero, we can rewrite
\begin{align*}
&\begin{aligned}[t]
\trace
\left(
  \transpose{\fadingTransformMatrix} \transpose{\messagesMatrix}
  \messagesMatrix \fadingTransformMatrix
\right)
\end{aligned}
\\
=\,
&\begin{aligned}[t]
&\trace
\left(
  \transpose{\fadingTransformMatrix} \transpose{\messagesMatrix}
  \messagesMatrix \fadingTransformMatrix
\right)
 \\
&+
\trace
\left(
  \transpose{\fadingTransformMatrix} \transpose{\messagesMatrix}
  \messagesMatrix \uncorrelatedApprox
\right)
-
\trace
\left(
  \transpose{(\uncorrelatedApprox)} \transpose{\messagesMatrix}
  \messagesMatrix \fadingTransformMatrix
\right)
\\
&-
\trace
\left(
  \transpose{(\uncorrelatedApprox)} \transpose{\messagesMatrix}
  \messagesMatrix \uncorrelatedApprox
\right)
+
\trace
\left(
  \transpose{(\uncorrelatedApprox)} \transpose{\messagesMatrix}
  \messagesMatrix \uncorrelatedApprox
\right)
\end{aligned}
\displaybreak[0] \\
=\,
&\begin{aligned}[t]
&\trace
\left(
  \transpose{(\fadingTransformMatrix - \uncorrelatedApprox)} \transpose{\messagesMatrix}
  \messagesMatrix (\fadingTransformMatrix + \uncorrelatedApprox)
\right)
\\
&+
\trace
\left(
  \transpose{(\uncorrelatedApprox)} \transpose{\messagesMatrix}
  \messagesMatrix \uncorrelatedApprox
\right)
\end{aligned}
\end{align*}
and use this together with (\ref{eq:hwEffectiveExpression}) to conclude
\begin{align}
\nonumber
\trace \hwEffectiveMatrix
=
&\trace
\left(
  \transpose{(\fadingTransformMatrix - \uncorrelatedApprox)} \transpose{\messagesMatrix}
  \messagesMatrix (\fadingTransformMatrix + \uncorrelatedApprox)
\right)
+
2\trace
\left(
  \transpose{\noiseTransformMatrix} \messagesMatrix
  \fadingTransformMatrix
\right)
\\
&+
\trace
\left(
  \transpose{(\uncorrelatedApprox)} \transpose{\messagesMatrix}
  \messagesMatrix \uncorrelatedApprox
\right)
+
\trace
\left(
  \transpose{\noiseTransformMatrix}
  \noiseTransformMatrix
\right).
\label{eq:hwMGF-trace}
\end{align}
Next, we argue that the terms in the last line are almost surely constant. For $\trace(\transpose{\noiseTransformMatrix}\noiseTransformMatrix)$ this is immediately clear. Moreover, we have
$
\trace{\left(
  \transpose{(\uncorrelatedApprox)} \transpose{\messagesMatrix}
  \messagesMatrix \uncorrelatedApprox
\right)}
=
\frobeniusnorm{\messagesMatrix\uncorrelatedApprox}^2.
$
We note that as per Remark~\ref{remark:uncorrelatedApprox} and using corresponding notation, we have $\messagesMatrix \uncorrelatedApprox=$
\begin{align*}
\begin{pmatrix}
\normalizedMessages(1) \uncorrelatedApprox^{(1)}\hspace{-7pt} &  0                      & 0      & \dots                                  & 0                                      \\
0                      & \normalizedMessages(1) \uncorrelatedApprox^{(2)}\hspace{-7pt} & 0      & \dots                                  & 0                                      \\                       &                        & \ddots &                                        &                                        \\
0                      & \dots                  & 0      &  \normalizedMessages(\numChannelUses)\uncorrelatedApprox^{(2\numChannelUses-1)} \hspace{-10pt} & 0                                      \\
0                      & \dots                  & 0      & 0                                      & \normalizedMessages(\numChannelUses) \uncorrelatedApprox^{(2\numChannelUses)}
\end{pmatrix}
\end{align*}
and because each $\uncorrelatedApprox^{(\indexChannelUses)}$ has only one nonzero entry per column, each entry of $\messagesMatrix \uncorrelatedApprox$ is the product of $U_\indexUsers(\indexChannelUses)$ with a deterministic term for some $\indexChannelUses, \indexUsers$ and therefore, its square can take only one value almost surely, and consequently, $\frobeniusnorm{\messagesMatrix\uncorrelatedApprox}^2$ also takes only one value almost surely.

We can use this in (\ref{eq:hwMGF-trace}) and incorporate the triangle inequality to obtain
\[
\subgaussnorm{\trace \hwEffectiveMatrix}
\leq
\subgaussnorm{
\subgaussvariable_1
}
+
2
\subgaussnorm{
\subgaussvariable_2
},
\]
where
\begin{align*}
\subgaussvariable_1
&:=
\trace
\left(
  \transpose{(\fadingTransformMatrix - \uncorrelatedApprox)} \transpose{\messagesMatrix}
  \messagesMatrix (\fadingTransformMatrix + \uncorrelatedApprox)
\right)
\\
\subgaussvariable_2
&:=
\trace
\left(
  \transpose{\noiseTransformMatrix} \messagesMatrix
  \fadingTransformMatrix
\right).
\end{align*}
To the end of bounding $\subgaussnorm{\subgaussvariable_1}$, we argue
\begin{align*}
\subgaussvariable_1
&=
\trace
\left(
  \messagesMatrix
  (\fadingTransformMatrix + \uncorrelatedApprox)
  \transpose{(\fadingTransformMatrix - \uncorrelatedApprox)}
  \transpose{\messagesMatrix}
\right)
\\
&\leq
\operatornorm{
  (\fadingTransformMatrix + \uncorrelatedApprox)
  \transpose{(\fadingTransformMatrix - \uncorrelatedApprox)}
}
\trace
\left(
  \messagesMatrix
  \transpose{\messagesMatrix}
\right)
\\
&\leq
2\numChannelUses\Delta(f||P)
\operatornorm{
  (\fadingTransformMatrix + \uncorrelatedApprox)
  \transpose{(\fadingTransformMatrix - \uncorrelatedApprox)}
}.
\end{align*}
The first inequality holds because for any square matrix $\generalmatrixOne$ and compatible column vector $\generalVector$, we have
\[
\transpose{\generalVector} (\operatornorm{\generalmatrixOne} \identityMatrix - \generalmatrixOne) \generalVector
=
\euclidnorm{\generalVector}^2
\left(
  \operatornorm{\generalmatrixOne}
  -
  \transpose{\left(\frac{\generalVector}{\euclidnorm{\generalVector}}\right)} \generalmatrixOne \frac{\generalVector}{\euclidnorm{\generalVector}}
\right)
\geq
0
\]
(see, e.g., \cite[Exercise I.2.10]{bhatia}) and therefore
$
\operatornorm{
  \generalmatrixOne
}
\identityMatrix
-
\generalmatrixOne
$
is positive semidefinite. The second inequality directly follows from (\ref{eq:correlated-QQT}). It follows, e.g., from \cite[Example 1.2]{buldygin}, that $\subgaussnorm{\subgaussvariable_1}$ is upper bounded by the first summand on the right hand side in (\ref{eq:hwMGF}).

In order to bound the sub-Gaussian norm of $\subgaussvariable_2$, we view it as a function of $(U_\indexUsers(\indexChannelUses))_{\indexUsers, \indexChannelUses = 1}^{\numUsers, \numChannelUses}$ and use part of the proof of the Bounded Differences Inequality~\cite[Theorem 6.2]{boucheron2013concentration} to bound the moment generating function. To this end, we define
\[
(\matrixOneEntry{\generalIndexOne}{\generalIndexTwo})_{\generalIndexOne', \generalIndexTwo'} =
\begin{cases}
1, &\generalIndexOne' = \generalIndexOne \text{ and } \generalIndexTwo' = \generalIndexTwo \\
0, &\text{otherwise.}
\end{cases}
\]
and note that a change in the value of $U_\indexUsers(\indexChannelUses)$ changes the value of
$\subgaussvariable_2$
by
\begin{align*}
&\hphantom{{}={}}
2 \sqrt{a_\indexUsers} \trace\left(
  \transpose{\noiseTransformMatrix}
  (
    \matrixOneEntry{2\indexChannelUses-1}{\numUsers(2\indexChannelUses-2)+\indexUsers}
    +
    \matrixOneEntry{2\indexChannelUses}{\numUsers(2\indexChannelUses-1)+\indexUsers}
  )
  \fadingTransformMatrix
\right)
\\
&=
2 \sqrt{a_\indexUsers} \trace\left(
  \fadingTransformMatrix
  \transpose{\noiseTransformMatrix}
  (
    \matrixOneEntry{2\indexChannelUses-1}{\numUsers(2\indexChannelUses-2)+\indexUsers}
    +
    \matrixOneEntry{2\indexChannelUses}{\numUsers(2\indexChannelUses-1)+\indexUsers}
  )
\right)
\displaybreak[0]
\\
&=
2 \sqrt{a_\indexUsers} \left(
  (
    \fadingTransformMatrix
    \transpose{\noiseTransformMatrix}
  )_{\numUsers(2\indexChannelUses-2)+\indexUsers,2\indexChannelUses-1}
  +
  (
    \fadingTransformMatrix
    \transpose{\noiseTransformMatrix}
  )_{\numUsers(2\indexChannelUses-1)+\indexUsers,2\indexChannelUses}
\right)
\\
&\leq
2 \sqrt{\powerconstraint} \left(
  (
    \fadingTransformMatrix
    \transpose{\noiseTransformMatrix}
  )_{\numUsers(2\indexChannelUses-2)+\indexUsers,2\indexChannelUses-1}
  +
  (
    \fadingTransformMatrix
    \transpose{\noiseTransformMatrix}
  )_{\numUsers(2\indexChannelUses-1)+\indexUsers,2\indexChannelUses}
\right)
\end{align*}
Following the proof of the Bounded Differences Inequality~\cite[Theorem 6.2]{boucheron2013concentration}, we can now conclude
\begin{align*}
\subgaussnorm{\subgaussvariable_2}^2
&\leq
\begin{aligned}[t]
\frac{1}{4}
\sum_{\indexUsers=1}^\numUsers
\sum_{\indexChannelUses=1}^\numChannelUses\Big(
  &2 \sqrt{\powerconstraint}
  (\fadingTransformMatrix\transpose{\noiseTransformMatrix})_{(2\indexChannelUses-2)\numUsers + \indexUsers, 2\indexChannelUses-1}
  \\
  &+
  2 \sqrt{\powerconstraint}
  (\fadingTransformMatrix\transpose{\noiseTransformMatrix})_{(2\indexChannelUses-1)\numUsers + \indexUsers, 2\indexChannelUses}
\Big)^2
\end{aligned}
\\
&\leq
\frac{1}{4}
\cdot
2
\cdot
4
\cdot
\powerconstraint
\frobeniusnorm{\fadingTransformMatrix\transpose{\noiseTransformMatrix}}^2,
\end{align*}
concluding the proof of the lemma.
\end{proof}

\begin{proof}[Proof of Theorem~\ref{th:correlated}]
What remains to be established is the concentration of $\vectorNorm{\rxVector}^2$ around its expectation. To this end, we observe
\begin{align}
\Probability\left(
  \absolute{\vectorNorm{\rxVector}^2 - \Expectation \vectorNorm{\rxVector}^2}
  \geq
  \tail
\right)
&=
\Probability\left(
  \absolute{
    \transpose{\randomBaseVector} \hwEffectiveMatrix \randomBaseVector
    -
    \Expectation (\transpose{\randomBaseVector} \hwEffectiveMatrix \randomBaseVector)
  }
  \geq
  \tail
\right)
\nonumber
\\
&\leq
\Probability\left(
  \absolute{\hwDiagonalSum} \geq \frac{\tail}{2}
\right)
+
\Probability\left(
  \absolute{\hwNondiagonalSum} \geq \frac{\tail}{2}
\right)
\label{eq:correlated-triangle}
\end{align}
where
\begin{align*}
\hwDiagonalSum
&:=
\sum_{\generalIndexOne=1}^{2\numUsers\numChannelUses + 2\numChannelUses} \left(
  \randomBaseVector_\generalIndexOne^2 \hwEffectiveMatrix_{\generalIndexOne, \generalIndexOne}
  -
  \Expectation\left(
    \randomBaseVector_\generalIndexOne^2 \hwEffectiveMatrix_{\generalIndexOne, \generalIndexOne}
  \right)
\right)
\displaybreak[0]
\\
\hwNondiagonalSum
&:=
\sum_{\substack{\generalIndexOne,\generalIndexTwo=1 \\ \generalIndexOne \neq \generalIndexTwo}}^{2\numUsers\numChannelUses + 2\numChannelUses}
  \randomBaseVector_\generalIndexOne \randomBaseVector_\generalIndexTwo \hwEffectiveMatrix_{\generalIndexOne, \generalIndexTwo}.
\end{align*}
We use Lemma~\ref{lemma:splitting}, Lemma~\ref{lemma:hwNormBound} and Lemma~\ref{lemma:hwMGF} to bound the moment generating function of $\hwDiagonalSum$ as
\begin{align*}
\Expectation\exp(\mgfVariable\hwDiagonalSum)
&\leq
\exp\left(
  \frac{\mgfVariable^2}{2}
  \left(
    \frac{8\generalSumUpperBound_1}{1-\auxVariable}
    +
    \generalSumUpperBound_2
  \right)
\right)
\\
&\leq
\exp\left(
  \frac{\mgfVariable^2}{2}
  \cdot
  \frac{8\generalSumUpperBound_1+\generalSumUpperBound_2}{1-\auxVariable}
\right)
\end{align*}
for any $\auxVariable \in (0,1)$ and $\mgfVariable \in (-c/(2\generalUpperBound)), c/(2\generalUpperBound))$, where
\begin{align*}
\generalUpperBound
&:=
\left(
  \sqrt{\Delta(f || \powerconstraint)} \operatornorm{\fadingTransformMatrix} + \operatornorm{\noiseTransformMatrix}
\right)^2
\displaybreak[0]
\\
\generalSumUpperBound_1
&:=
\generalUpperBound
\left(
  \sqrt{\Delta(f || \powerconstraint)} \frobeniusnorm{\fadingTransformMatrix} + \frobeniusnorm{\noiseTransformMatrix}
\right)^2
\displaybreak[0]
\\
\generalSumUpperBound_2
&:=
\begin{multlined}[t]
\Big(
  4 \numChannelUses \Delta(f||\powerconstraint)
  \operatornorm{(\fadingTransformMatrix+\uncorrelatedApprox)\transpose{(\fadingTransformMatrix-\uncorrelatedApprox)}}
  \\
  + 2\sqrt{2 \powerconstraint} \frobeniusnorm{\fadingTransformMatrix\transpose{\noiseTransformMatrix}}
\Big)^2
\end{multlined}
\end{align*}
By Lemma~\ref{lemma:app-4}, this yields
\begin{align*}
\Probability\left(\absolute{\hwDiagonalSum} \geq \frac{\tail}{2}\right)
\leq
2\exp\left(
  -(1-\auxVariable) \frac{\tail^2}{64\generalSumUpperBound_1+8\generalSumUpperBound_2}
\right)
\end{align*}
in case $0 < \tail \leq \frac{\auxVariable}{1-\auxVariable} \cdot \frac{8 \generalSumUpperBound_1 + \generalSumUpperBound_2}{\generalUpperBound}$ and
\begin{align*}
\Probability\left(\absolute{\hwDiagonalSum} \geq \frac{\tail}{2}\right)
\leq
2\exp\left(
  -\frac{\auxVariable \tail}{8\generalUpperBound}
\right)
\end{align*}
otherwise. Since the first case term is increasing with $\auxVariable$ and the second case term is decreasing, the optimal value for $\auxVariable$ is where the two cases meet, which is at
\begin{align*}
\auxVariable = \frac{\generalUpperBound\tail}{\generalUpperBound\tail + 8\generalSumUpperBound_1 + \generalSumUpperBound_2}.
\end{align*}
Substituting this, we get
\begin{align*}
\Probability\left(\absolute{\hwDiagonalSum} \geq \frac{\tail}{2}\right)
\leq
2\exp\left(
  -\frac{\tail^2}{64\generalSumUpperBound_1 + 8\generalSumUpperBound_2 + 8\generalUpperBound\tail}
\right).
\end{align*}
Turning our attention to $\hwNondiagonalSum$, we note that by~\cite[Theorem 2.1]{bhatia1989comparing} the operator norm of the off-diagonal part of $\hwEffectiveMatrix$ can be upper bounded by $2\operatornorm{\hwEffectiveMatrix}$ and thus by $2\generalUpperBound$. Therefore, we can directly apply Lemma~\ref{lemma:hw-offdiagonal} and get
\begin{align*}
\Probability\left(\absolute{\hwNondiagonalSum} \geq \frac{\tail}{2}\right)
\leq
2\exp\left(
-
\frac{\tail^2}{1024 \generalSumUpperBound_1 + 64 \generalUpperBound\tail}
\right).
\end{align*}
Substituting these into (\ref{eq:correlated-triangle}) and using (\ref{eq:correlated-fY}) concludes the proof.
\end{proof}

\section{Conclusion}
In this paper, we have proposed and analyzed a scheme for analog \gls{ota} computation in a very general setting of fast-fading channels with no \gls{csi}. The derived nonasymptotical bounds are valid for large classes of fading and noise that can have non-Gaussian and correlated components. They work regardless of how the sources are distributed and in particular allow for arbitrary correlations between them. We have given examples for how the scheme can be applied to distributed \gls{ml} problems and demonstrated the suitability with extensive numerical simulations for the special case of binary classification problems.

Promising questions for future research include:
\begin{itemize}
 \item Can the error bounds of Theorem~\ref{th:correlated} be sharpened?
 \item Can the results of this work be applied to derive theoretical error guarantees for training in Horizontal \gls{fl}?
 \item Are there more efficient decentralized training algorithms for \gls{vfl} available taking advantage of \gls{ota} computation?
 \item In~\cite{amiri2019collaborative}, an analog \gls{ota} computation scheme that takes advantage of \gls{mimo} in the case of \gls{csi} at the receiver has been proposed and analyzed asymptotically. The authors of~\cite{goldenbaum2014channel} have explored the case of multiple receive antennas and (among other cases) no \gls{csi} at the transmitter or receiver and provided asymptotic error bounds. It would be interesting to consider the case of no \gls{csi} in a \gls{mimo} system and to find nonasymptotic characterizations of the performance gain that can be achieved.
 \item \cite{wang2020wireless} proposes a way of optimizing the performance of \gls{ota} computation schemes by means of channel shaping via \glspl{irs}. The error bounds derived in this work quantify the performance of \gls{ota} computation schemes in dependence of the correlation properties of the wireless channel. With this knowledge, it would be interesting to investigate how we could use \glspl{irs} in such a way that we optimize the performance of the \gls{ota} computation.
\end{itemize}

\bibliographystyle{IEEEtran}
\bibliography{references}

\clearpage

%
%
\renewcommand\appendixname{Supplemental Material}
\appendix

\subsection{Preliminaries on \rev{Sub-Gaussian} and Sub-Exponential Random Variables}
\label{app:sub-exp-sub-gauss}
We begin with a definition that is adapted from~\cite[Definition 3.4.1]{vershynin}. For $\reals^\rvDimension$-valued random variables $\generalRandomVector$, we define the \rev{sub-Gaussian} norm as
\begin{multline}
\subgaussnorm{\generalRandomVector}
:=
\inf\Big\{
  \subgaussdefInfVar
  :
  \forall \subgaussdefUnitVector \in \sphere{\rvDimension}~
  \forall \mgfVariable \in \reals~
  \\
  \Expectation \exp(\mgfVariable\innerprod{\generalRandomVector}{\subgaussdefUnitVector})
  \leq
  \exp\left(\frac{\subgaussdefInfVar^2\mgfVariable^2}{2}\right)
\Big\}
\end{multline}
and we observe that if all entries of $\generalRandomVector$ have a \rev{sub-Gaussian} norm bounded by $\subgaussBound$ and are independent, we have for any $\subgaussdefUnitVector \in \sphere{\rvDimension}$
\begin{align*}
\Expectation \exp\left(\mgfVariable \innerprod{\generalRandomVector}{\subgaussdefUnitVector}\right)
&=
\Expectation \exp\left(\mgfVariable \sum_{\rvDimIndex=1}^\rvDimension \generalRandomVector_\rvDimIndex \subgaussdefUnitVector_\rvDimIndex\right)
\displaybreak[0]\\
&=
\prod_{\rvDimIndex=1}^\rvDimension \Expectation \exp\left(\mgfVariable \generalRandomVector_\rvDimIndex \subgaussdefUnitVector_\rvDimIndex\right)
\displaybreak[0]\\
&\leq
\prod_{\rvDimIndex=1}^\rvDimension \exp\left(\frac{\subgaussBound^2\subgaussdefUnitVector_\rvDimIndex^2\mgfVariable^2}{2}\right)
=
\exp\left(\frac{\subgaussBound^2 \mgfVariable^2}{2}\right)
\end{align*}
and therefore $\subgaussnorm{\generalRandomVector} \leq \subgaussBound$.

In the following, we recall some basic definitions and results from \cite[Chapter 1]{buldygin}. For a  random variable $X$ we define\footnote{Note that as with our definition of the \rev{sub-Gaussian} norm, other norms on the space of sub-exponential random variables that appear in the literature are equivalent to $\subexpnorm{\cdot}$ (see, e.g.,~\cite{buldygin}). The particular definition we choose here matters, however, because we want to derive results in which no unspecified constants appear.}
\begin{equation}\label{eq:sub-exp-norm-def}
\subexpnorm{X} := \sup\limits_{k \ge 1} \left( \frac{\mathbb{E}(|X|^k)}{k!} \right)^\frac{1}{k}
\end{equation}
If $ \subexpnorm{X}< \infty  $ then $X$ is called a sub-exponential random variable. $\subexpnorm{\cdot}$ defines a semi-norm on the vector space of sub-exponential 
random variables \cite[Remark 1.3.2]{buldygin}. Typical examples of sub-exponential random variables are bounded random variables and random variables with exponential distribution.
We collect some useful  properties of and interrelations between the sub-exponential and \rev{sub-Gaussian} norms in the following lemma.
%
%
\begin{lemma}\label{lemma:sub-exp-properties}
Let $X,Y$ be random variables. Then:
\begin{enumerate}
\item If $X$ is $\mathcal{N}(\mu, \sigma^2)$ then we have
\begin{equation}\label{eq:gaussian-bound}
\subgaussnorm{X} = \sigma.
\end{equation}
\item (Rotation Invariance) If $X_1, \dots, X_M$ are independent, \rev{sub-Gaussian} and centered, we have
\begin{equation}\label{eq:subgauss-rotinv}
\subgaussnorm{\sum\limits_{m=1}^M X_m}^2 \leq \sum\limits_{m=1}^M \subgaussnorm{X_m}^2
\end{equation}
\item If $X$ is a random variable with $| X |\le 1$ with probability $1$ and if $Y$ is independent of $X$ and \rev{sub-Gaussian} then we have
\begin{equation}\label{eq:bounded-12-bound}
\subgaussnorm{X \cdot Y}\le \subgaussnorm{Y}.
\end{equation}
\item If $X$ and $Y$ are \rev{sub-Gaussian} and centered, then $X\cdot Y$ is sub-exponential and
\begin{equation}\label{eq:c-s}
\subexpnorm{X \cdot Y} \le 2 \cdot \subgaussnorm{X} \cdot \subgaussnorm{Y}.
\end{equation}
\item (Centering) If $X$ is sub-exponential and $X \geq 0$ almost surely, then
\begin{equation}\label{eq:centering}
\subexpnorm{X- \mathbb{E}(X)} \leq \subexpnorm{X}.
\end{equation}
\end{enumerate}
\end{lemma}
\begin{proof}
(\ref{eq:gaussian-bound}) follows in a straightforward fashion by calculating the moment generating function of $X$. (\ref{eq:subgauss-rotinv}) is e.g. proven in \cite[Lemma 1.1.7]{buldygin}. (\ref{eq:bounded-12-bound}) follows directly from the definition conditioning on $X$. We show (\ref{eq:c-s}) first for $X=Y$. In this case, we have
\begin{align*}
\subexpnorm{X^2} &= \sup_{k \geq 1} \left(\frac{\mathbb{E}X^{2k}}{k!}\right)^{\frac{1}{k}} \leq \sup_{k \geq 1} \left( \frac{2^{k+1}k^k\subgaussnorm{X}^{2k}}{e^k k!} \right)^\frac{1}{k} \nonumber \\
&= 2 \subgaussnorm{X}^2 \sup_{k \geq 1} \left(\frac{2^\frac{1}{k} k}{e (k!)^\frac{1}{k}}\right) \leq 2 \subgaussnorm{X}^2,
\end{align*}
where the first inequality is by \cite[Lemma 1.1.4]{buldygin} and the second follows from $2k^k/k! \leq e^k$, which is straightforward to prove for $k\ge 1$ by induction. In the general case, we have
\pagebreak[0]
\begin{align*}
\subexpnorm{XY} &= \subgaussnorm{X}\subgaussnorm{Y}\subexpnorm{\frac{XY}{\subgaussnorm{X}\subgaussnorm{Y}}}\displaybreak[0]\\
&\leq \subgaussnorm{X}\subgaussnorm{Y} \subexpnorm{\frac{1}{2} \left(\frac{X}{\subgaussnorm{X}}\right)^2 + \frac{1}{2} \left(\frac{Y}{\subgaussnorm{Y}}\right)^2 } \\
&\leq 2 \subgaussnorm{X}\subgaussnorm{Y},
\end{align*}
where the first inequality can be verified in (\ref{eq:sub-exp-norm-def}), considering that $ab \leq a^2/2 + b^2/2$ for all $a,b \in \reals$, and the second inequality follows from the triangle inequality and the special case $X=Y$.

For (\ref{eq:centering}), we assume without loss of generality $\mathbb{E} X = 1$ (otherwise we can scale $X$), and note that for all $a \in [0,\infty)$ and $k \geq 1$, $a^k-|a-1|^k > a-1$ and thus
$
\mathbb{E}(X^k - |X-1|^k) \geq \mathbb{E}(X-1) = 0.
$
\end{proof}
\subsection{Proof of Lemma~\ref{lemma:hw-offdiagonal}}
\label{appendix-hw}
The proof closely follows the proof of the Hanson-Wright inequality in~\cite[Theorem 6.2.1]{vershynin}. We carry out the changes that are necessary to arrive at explicit constants. To this end, we begin with some slightly modified versions of lemmas used as ingredients in the proof of Bernstein's inequality in~\cite[Theorem 1.5.2]{buldygin}.
%
%
%
\begin{lemma}\label{lemma:app-1}
Let $X$ be a random variable with $\E(X)=0$ and $\subexpnorm{X}< +\infty$. For any $\lambda \in\reals$ with $| \lambda \subexpnorm{X} |<1$ we have
\begin{equation*}
\E (\exp(\lambda X))\le 1+ |\lambda|^2 \subexpnorm{X}^2 \cdot \frac{1}{1- |\lambda \subexpnorm{X}|}.
\end{equation*}
\end{lemma}
\begin{proof} Let $\lambda \in \reals$ satisfy $| \lambda \subexpnorm{X} |<1$.
Then
\begin{align}
&\hphantom{{}={}}
\E (\exp(\lambda X))= 1+ \sum_{k=2}^{\infty} \frac{\lambda ^k \E(X^k)}{k!}\nonumber \displaybreak[0] \\
&\le 1+ \sum_{k=2}^{\infty} \frac{|\lambda| ^k \E(|X|^k)}{k!}\nonumber 
\le 1+ \sum_{k=2}^{\infty} |\lambda|^k \subexpnorm{X}^k\nonumber \displaybreak[0] \\
&= 1+ |\lambda |^2 \subexpnorm{X}^2\left ( \sum_{k=0}^{\infty} |\lambda \subexpnorm{X}  |^k \right)\nonumber\\
&= 1+ |\lambda|^2 \subexpnorm{X}^2 \cdot \frac{1}{1- |\lambda \subexpnorm{X}|},
\end{align}
where in the last line we have used $| \lambda \subexpnorm{X} |<1$.
\end{proof}
In the next lemma we derive an exponential bound depending on $\subexpnorm{X}$ on the moment generating function of the random variable $X$.
%
%
%
\begin{lemma}\label{lemma:app-2}
Let $X$ be a random variable with $\E(X)=0$ and $\subexpnorm{X}< +\infty$. For any $c\in (0,1)$ and 
$\lambda \in \left (  -\frac{c}{\subexpnorm{X}} ,  \frac{c}{\subexpnorm{X}} \right)$ we have
\begin{equation*}
\E (\exp(\lambda X))\le \exp\Bigl(\frac{\lambda^2}{2}\frac{2 \cdot  \subexpnorm{X}^2}{1-c}\Bigr).
\end{equation*}
\end{lemma}
\begin{proof}
For $\lambda \in \left (  -\frac{c}{\subexpnorm{X}} ,  \frac{c}{\subexpnorm{X}} \right)$ we have
\begin{equation}\label{eq:app-1}
| \lambda \subexpnorm{X}    |< c<1,
\end{equation}
therefore by Lemma \ref{lemma:app-1}
\begin{align*}
&\hphantom{{}={}}
\E (\exp(\lambda X)) \le 1+ |\lambda|^2 \subexpnorm{X}^2 \cdot \frac{1}{1- |\lambda \subexpnorm{X}|}\nonumber\displaybreak[0]\\
&\le  1+ |\lambda|^2 \subexpnorm{X}^2 \cdot \frac{1}{1- c}
\le \exp\Bigl( \frac{\lambda^2}{2}\frac{2 \cdot \subexpnorm{X}^2}{1-c}     \Bigr),
\end{align*}
where in the second line we have used the first inequality in (\ref{eq:app-1})  and the last line is by the numerical inequality $1+x \le \exp (x)$ valid for $x\ge 0$.
\end{proof}
%
%
%
\begin{lemma}\label{lemma:app-3}
Let $X_1, \ldots , X_M$ be independent random variables with $\E(X_i)=0$ and $\subexpnorm{X_i}<+\infty$, $i=1,\ldots ,M$. 
Let $L:=\max_{1\le i\le M} \subexpnorm{X_i}$, $c\in (0,1)$, and $\lambda \in \left (  -\frac{c}{L} ,  \frac{c}{L} \right)$.
Then for $S_M:= \sum_{i=1}^M X_i$ we have
\begin{equation}\label{eq:app-2}
\E (\exp(\lambda S_M))\le \exp \Bigl(  \frac{\lambda^2}{2}\frac{2 \cdot \sum_{i=1}^M \subexpnorm{X_i}^2}{1-c}  \Bigr).
\end{equation}
\end{lemma}
\begin{proof}
By independence of $X_1, \ldots , X_M$ we have
\begin{equation*}
\E (\exp(\lambda S_M))= \prod_{i=1}^M \E (\exp(\lambda X_i)).
\end{equation*}
Combining this with Lemma \ref{lemma:app-2} proves the lemma.
\end{proof}
The next lemma establishes the basic tail bound for random variables satisfying inequalities of type (\ref{eq:app-2}).
The proof can be found in \cite[Lemma 1.4.1]{buldygin}.
%
%
%
\begin{lemma}\label{lemma:app-4}
Let $X$ be a random variable with $\E(X)=0$. If there exist $\tau\ge 0$ and $\Lambda >0$ such that
\begin{equation*}
\E (\exp(\lambda X))\le \exp\Bigl(\frac{\lambda^2}{2}\tau^2\Bigr),
\end{equation*}
holds for all $\lambda\in (-\Lambda, \Lambda)$,
then for any $t\ge 0$ we have
\begin{equation*}
\mathbb{P}(|X |\ge t)\le 2 \cdot Q(t),
\end{equation*}
where
\begin{equation*}
Q(t)=  
\begin{cases}   \exp \left (- \frac{t^2}{2 \tau^2}  \right ), & 0< t\le \Lambda \tau^2 \\
         \exp\left (     - \frac{\Lambda t}{2}    \right), & \Lambda  \tau^2\le t .
\end{cases}
\end{equation*}
\end{lemma}

The following lemma is a slightly modified version of~\cite[Lemma 6.2.3]{vershynin}.
\begin{lemma}[Comparison Lemma]
\label{lemma:comparison}
Let $\generalRandomVector$ and $\generalRandomVector'$ be independent, $\reals^\rvDimension$-valued, centered and \rev{sub-Gaussian} random variables, and let $\gaussianRandomVector, \gaussianRandomVector'$ be independent and distributed according to $\mvNormal{0}{\identityMatrix_\rvDimension}$. Let $\generalTransformMatrix \in \reals^{\rvDimension \times \rvDimension}$ and $\mgfVariable \in \reals$. Then
\begin{align*}
\Expectation \exp(
  \mgfVariable \transpose{\generalRandomVector} \generalTransformMatrix \generalRandomVector'
)
\leq
\Expectation \exp(
  \mgfVariable \subgaussnorm{\generalRandomVector} \subgaussnorm{\generalRandomVector'} \transpose{\gaussianRandomVector} \generalTransformMatrix \gaussianRandomVector'.
)
\end{align*}
\end{lemma}
\begin{proof}
We first observe that for any $\generalDeterministicVector \in \reals^\rvDimension$,
\begin{align}
&\hphantom{{}={}}
\Expectation(
  \exp(\mgfVariable\innerprod{\generalRandomVector}{\generalDeterministicVector})
)
=
\Expectation\left(
  \exp\left(
    \mgfVariable
    \euclidnorm{\generalDeterministicVector}
    \innerprod{\generalRandomVector}{\frac{\generalDeterministicVector}{\euclidnorm{\generalDeterministicVector}}}
  \right)
\right)
\displaybreak[0]
\nonumber
\\
&\leq
\exp\Bigl(
  \frac{
         \mgfVariable^2
         \euclidnorm{\generalDeterministicVector}^2
         \subgaussnorm{\generalRandomVector}^2
       }
       {2}
\Bigr)
=
\Expectation\Bigl(
  \exp\bigl(
    \mgfVariable
    \subgaussnorm{\generalRandomVector}
    \innerprod{\gaussianRandomVector}{\generalDeterministicVector}
  \bigr)
\Bigr),
\label{eq:comparison-gaussianmgf}
\end{align}
where the inequality in (\ref{eq:comparison-gaussianmgf}) is by the
definition of vector-valued \rev{sub-Gaussian} random variables and the
equality is obtained by calculating the moment-generation function of $\innerprod{\gaussianRandomVector}{\generalDeterministicVector}$.
We can now conclude the proof from the following:
\begin{align}
&\hphantom{{}={}}
\Expectation(\exp(\mgfVariable \transpose{\generalRandomVector} \generalTransformMatrix \generalRandomVector'))
\nonumber\\
&=
\label{eq:comparison-fubini1}
\Expectation_{\generalRandomVector'} (
  \Expectation_{\generalRandomVector} (
    \exp(\mgfVariable \innerprod{\generalRandomVector}{\generalTransformMatrix \generalRandomVector'})
  )
)
\displaybreak[0]
\\
&\leq
\label{eq:comparison-obs1}
\Expectation_{\generalRandomVector'} (
  \Expectation_{\gaussianRandomVector} (
    \exp(\mgfVariable \subgaussnorm{\generalRandomVector} \innerprod{\gaussianRandomVector}{\generalTransformMatrix \generalRandomVector'})
  )
)
\displaybreak[0]
\\
&=
\label{eq:comparison-fubini2}
\Expectation_{\gaussianRandomVector} (
  \Expectation_{\generalRandomVector'} (
    \exp(\mgfVariable \subgaussnorm{\generalRandomVector} \innerprod{\generalRandomVector'}{\transpose{\generalTransformMatrix}\gaussianRandomVector})
  )
)
\displaybreak[0]
\\
&\leq
\label{eq:comparison-obs2}
\Expectation_{\gaussianRandomVector} (
  \Expectation_{\gaussianRandomVector'} (
    \exp(\mgfVariable \subgaussnorm{\generalRandomVector} \subgaussnorm{\generalRandomVector'} \innerprod{\gaussianRandomVector'}{\transpose{\generalTransformMatrix}\gaussianRandomVector})
  )
)
\\
&=
\label{eq:comparison-fubini3}
\Expectation (
  \exp(\mgfVariable \subgaussnorm{\generalRandomVector} \subgaussnorm{\generalRandomVector'} \transpose{\gaussianRandomVector} \generalTransformMatrix \gaussianRandomVector')
),
\end{align}
where (\ref{eq:comparison-fubini1}), (\ref{eq:comparison-fubini2}) and (\ref{eq:comparison-fubini3}) are due to Fubini's theorem and elementary transformations and (\ref{eq:comparison-obs1}) and (\ref{eq:comparison-obs2}) are both instances of the observation (\ref{eq:comparison-gaussianmgf}).
\end{proof}

\begin{proof}[Proof of Lemma~\ref{lemma:hw-offdiagonal}]
We can write
\begin{align}
\transpose{\generalRandomVector} \generalTransformMatrix \generalRandomVector
=
\sum_{\substack{\rvDimIndex, \rvDimIndexAlt = 1, \rvDimIndex \neq \rvDimIndexAlt}}^\rvDimension
  \generalRandomVector_\rvDimIndex
  \generalTransformMatrix_{\rvDimIndex,\rvDimIndexAlt}
  \generalRandomVector_\rvDimIndexAlt,
\end{align}
and since $\generalRandomVector$ is centered, $\Expectation\left(\transpose{\generalRandomVector} \generalTransformMatrix \generalRandomVector\right) = 0$ immediately follows. Let $\generalRandomVector'$ be an independent copy of $\generalRandomVector$, and let $\gaussianRandomVector$ and $\gaussianRandomVector'$ be independently distributed according to $\mvNormal{0}{\identityMatrix_\rvDimension}$.
We denote the singular values of $\generalTransformMatrix$ with $\generalTransformMatrixSingular_1, \dots, \generalTransformMatrixSingular_\rvDimension$. With these definitions, we bound the moment-generating function of $\transpose{\generalRandomVector} \generalTransformMatrix \generalRandomVector$ as
\begin{align}
\Expectation \exp\left(\mgfVariable \transpose{\generalRandomVector} \generalTransformMatrix \generalRandomVector \right)
&=
\Expectation \exp\left(\mgfVariable \transpose{\generalRandomVector} \generalTransformMatrix \generalRandomVector\right)
\displaybreak[0]
\\
&\leq
\label{eq:hw-decoupling}
\Expectation \exp\left(4 \mgfVariable \transpose{\generalRandomVector} \generalTransformMatrix \generalRandomVector'\right)
\displaybreak[0]
\\
&\leq
\label{eq:hw-comparison}
\Expectation \exp\left(4 \mgfVariable \subgaussnorm{\generalRandomVector}^2 \transpose{\gaussianRandomVector} \generalTransformMatrix \gaussianRandomVector'\right)
\displaybreak[0]
\\
&=
\label{eq:hw-singular-transform}
\Expectation \exp\Bigl(4 \mgfVariable \subgaussnorm{\generalRandomVector}^2 \sum_{\rvDimIndex=1}^\rvDimension \gaussianRandomVectorTransformed_\rvDimIndex \gaussianRandomVectorTransformed'_\rvDimIndex \generalTransformMatrixSingular_\rvDimIndex\Bigr)
\displaybreak[0]
\\
&\leq
\label{eq:hw-calculatemgf}
\exp\Bigl(
  \frac{\mgfVariable^2}
       {2}
  \cdot
  \frac{128 \subgaussnorm{\generalRandomVector}^4 \sum_{\rvDimIndex=1}^\rvDimension \generalTransformMatrixSingular_\rvDimIndex^2}
       {1-\auxVariable}
\Bigr),
\end{align}
where (\ref{eq:hw-decoupling}) is due to the Decoupling
Theorem~\cite[Theorem 6.1.1]{vershynin}, (\ref{eq:hw-comparison}) is
an application of Lemma~\ref{lemma:comparison},
(\ref{eq:hw-singular-transform}) holds for suitably transformed
versions $\gaussianRandomVectorTransformed,
\gaussianRandomVectorTransformed'$ of $\gaussianRandomVector,
\gaussianRandomVector'$ (note that they are still independent and
follow the same distribution) and (\ref{eq:hw-calculatemgf}) is true
if $\auxVariable \in (0,1)$ and $\abs{\mgfVariable} < \auxVariable /
(8 \subgaussnorm{\generalRandomVector}^2 \max_{1 \leq \rvDimIndex \leq
  \rvDimension} \generalTransformMatrixSingular_\rvDimIndex)$
according to Lemma~\ref{lemma:app-3}. So we can apply
Lemma~\ref{lemma:app-4} to obtain
\begin{align}
\label{eq:hw-auxtailprob}
\Probability\left(
  \absolute{
    \transpose{\generalRandomVector} \generalTransformMatrix \generalRandomVector
  }
  \geq
  \tail
\right)
\leq
2\exp\Bigl(
  -\frac{\tail^2(1-\auxVariable)}
        {256 \subgaussnorm{\generalRandomVector}^4 \sum_{\rvDimIndex=1}^\rvDimension \generalTransformMatrixSingular_\rvDimIndex^2}
\Bigr)
\end{align}
in case $\tail \leq \frac{\auxVariable}{1-\auxVariable} \cdot \frac{16\subgaussnorm{\generalRandomVector}^2 \sum_{\rvDimIndex=1}^\rvDimension \generalTransformMatrixSingular_\rvDimIndex^2}{\max_{1 \leq \rvDimIndex \leq \rvDimension} \generalTransformMatrixSingular_\rvDimIndex}$ and
\begin{align*}
\Probability\left(
  \absolute{
    \transpose{\generalRandomVector} \generalTransformMatrix \generalRandomVector
  }
  \geq
  \tail
\right)
\leq
2\exp\Bigl(
  -c \cdot \frac{\tail}{16 \subgaussnorm{\generalRandomVector}^2 \max\nolimits_{1 \leq \rvDimIndex \leq \rvDimension} \generalTransformMatrixSingular_\rvDimIndex}
\Bigr)
\end{align*}
otherwise. We next choose $\auxVariable$ so as to minimize the upper bound on the tail probability. Because the bound in the first case is increasing with $\auxVariable$ while it is decreasing in the second case, the optimal choice for $\auxVariable$ is where the two cases meet. We can therefore calculate the optimal $\auxVariable$ as
\begin{align*}
\auxVariable
=
\frac{\tail \max\nolimits_{1 \leq \rvDimIndex \leq \rvDimension} \generalTransformMatrixSingular_\rvDimIndex}
     {\tail \max\nolimits_{1 \leq \rvDimIndex \leq \rvDimension} \generalTransformMatrixSingular_\rvDimIndex + 16 \subgaussnorm{\generalRandomVector}^2 \sum_{\rvDimIndex=1}^\rvDimension \generalTransformMatrixSingular_\rvDimIndex^2}
\end{align*}
and substituting this in (\ref{eq:hw-auxtailprob}), we obtain
\begin{multline*}
\Probability\bigl(
  \absolute{
    \transpose{\generalRandomVector} \generalTransformMatrix \generalRandomVector
  }
  \geq
  \tail
\bigr)
\leq
\\
2\exp\Bigl(
  -
  \frac{\tail^2}
       {16 \tail \subgaussnorm{\generalRandomVector}^2 \max\limits_{1 \leq \rvDimIndex \leq \rvDimension} \generalTransformMatrixSingular_\rvDimIndex + 256 \subgaussnorm{\generalRandomVector}^4 \sum_{\rvDimIndex=1}^\rvDimension \generalTransformMatrixSingular_\rvDimIndex^2}
\Bigr).
\end{multline*}
The bounds $\subgaussnorm{\generalRandomVector} \leq \subgaussBound$, $\abs{\generalTransformMatrixSingular_\rvDimIndex} \leq \operatornorm{\generalTransformMatrix}$, and identity $\frobeniusnorm{\generalTransformMatrix}^2 = \sum_{\rvDimIndex=1}^\rvDimension \generalTransformMatrixSingular^2$ allow us to conclude the proof of the lemma.
\end{proof}
\end{document}